\documentclass[11pt]{article}
\usepackage{microtype}
\usepackage{makeidx}
\usepackage[backref=page]{hyperref}

\usepackage{euscript}
\usepackage{amsmath,amsfonts,amsthm}
\usepackage{dsfont}
\usepackage{latexsym}
\usepackage{subfigure}
\usepackage{graphicx}
\usepackage{fancybox}
\usepackage{fullpage}
\usepackage{paralist}
\usepackage{color}
\usepackage{wrapfig}
\usepackage{tikz}
\usetikzlibrary{decorations.pathreplacing}
\usepackage{setspace}
\usepackage{algorithm}
\usepackage[noend]{algpseudocode}
\usepackage[framemethod=tikz]{mdframed}
\usepackage{xspace}
\usepackage{framed}
\usepackage{thmtools}
\usepackage{thm-restate}
\usepackage{tabu}

\newtheorem{theorem}{Theorem}[section]
\newtheorem{corollary}[theorem]{Corollary}
\newtheorem{lemma}[theorem]{Lemma}

\newtheorem{definition}[theorem]{Definition}
\newtheorem{remark}[theorem]{Remark}

\newtheorem{fact}[theorem]{Fact}

\newtheorem{framework}[theorem]{Framework}

\newenvironment{proofof}[1]{\begin{trivlist} \item {\bf Proof
#1:~~}}
  {\qed\end{trivlist}}

\newcommand{\namedref}[2]{\hyperref[#2]{#1~\ref*{#2}}}
\newcommand{\thmlab}[1]{\label{thm:#1}}
\newcommand{\thmref}[1]{\namedref{Theorem}{thm:#1}}
\newcommand{\lemlab}[1]{\label{lem:#1}}
\newcommand{\lemref}[1]{\namedref{Lemma}{lem:#1}}

\newcommand{\corlab}[1]{\label{cor:#1}}
\newcommand{\corref}[1]{\namedref{Corollary}{cor:#1}}
\newcommand{\seclab}[1]{\label{sec:#1}}

\newcommand{\applab}[1]{\label{app:#1}}
\newcommand{\appref}[1]{\namedref{Appendix}{app:#1}}
\newcommand{\factlab}[1]{\label{fact:#1}}
\newcommand{\factref}[1]{\namedref{Fact}{fact:#1}}

\newcommand{\figlab}[1]{\label{fig:#1}}
\newcommand{\figref}[1]{\namedref{Figure}{fig:#1}}
\newcommand{\alglab}[1]{\label{alg:#1}}
\renewcommand{\algref}[1]{\namedref{Algorithm}{alg:#1}}

\newcommand{\frameref}[1]{\namedref{Framework}{frame:#1}}
\newcommand{\framelab}[1]{\label{frame:#1}}
\newcommand{\equ}{\textsc{equality}}
\newcommand{\eq}{\textsc{eq}}
\newcommand{\fail}{\textbf{FAIL}}

\def \H    {\mdef{\mathcal{H}}}
\def \v    {\mdef{\mathbf{v}}}
\def \x    {\mdef{\mathbf{x}}}
\def \M    {\mdef{\mathbf{M}}}
\def \m    {\mdef{\mathbf{m}}}
\def \A    {\mdef{\mathcal{A}}}

\def \R    {\mdef{\mathbb{R}}}
\def \S    {\mdef{\mathcal{S}}}
\def \T    {\mdef{\mathcal{T}}}

\def \sampler    {\mdef{\mathsf{Sampler}}}

\def \estimate    {\mdef{\mathsf{Estimate}}}
\def \misragries    {\mdef{\mathsf{MisraGries}}}
\def \zsampler    {\mdef{\mathsf{L0Sampler}}}

\newcommand\norm[1]{\left\lVert#1\right\rVert}
\newcommand{\PPr}[1]{\ensuremath{\mathbf{Pr}\left[#1\right]}}
\newcommand{\Ex}[1]{\ensuremath{\mathbb{E}\left[#1\right]}}

\renewcommand{\O}[1]{\ensuremath{\mathcal{O}\left(#1\right)}}
\newcommand{\tO}[1]{\ensuremath{\tilde{\mathcal{O}}\left(#1\right)}}
\newcommand{\eps}{\epsilon}
\newcommand{\cP}{\mathcal{P}}

\newcommand{\mdef}[1]{{\ensuremath{#1}}\xspace}  % Math Def which can also be used in normal text.
\DeclareMathOperator*{\polylog}{polylog}
\DeclareMathOperator*{\poly}{poly}
\DeclareMathOperator*{\cost}{cost}
\DeclareMathOperator*{\out}{out}
\DeclareMathOperator*{\reff}{ref}
\DeclareMathOperator*{\rcost}{rcost}
\DeclareMathOperator*{\rerr}{rerr}
\DeclareMathOperator*{\verr}{verr}
\newcommand{\superscript}[1]{\ensuremath{^{\mbox{\tiny{\textit{#1}}}}}\xspace}
\def \th {\superscript{th}}     % 'The i-th entry it a list...' --> i\th
\newcommand{\ignore}[1]{}

\newif\ifnotes\notestrue %set this to true if notes are visible and to false (next line) if they should be hidden
\ifnotes
\newcommand{\samson}[1]{\textcolor{purple}{{\bf (Samson:} {#1}{\bf ) }} \marginpar{\tiny\bf
             \begin{minipage}[t]{0.5in}
               \raggedright S:
            \end{minipage}}}            							
\else
\newcommand{\samson}[1]{}
\fi

\hypersetup{
     colorlinks   = true,
     citecolor    = blue,
		 linkcolor		= red
}

\makeatletter
\renewcommand*{\@fnsymbol}[1]{\textcolor{mahogany}{\ensuremath{\ifcase#1\or *\or \dagger\or \ddagger\or
 \mathsection\or \triangledown\or \mathparagraph\or \|\or **\or \dagger\dagger
   \or \ddagger\ddagger \else\@ctrerr\fi}}}
\makeatother

\providecommand{\email}[1]{\href{mailto:#1}{\nolinkurl{#1}\xspace}}

\definecolor{mahogany}{rgb}{0.75, 0.25, 0.0}
\definecolor{darkblue}{rgb}{0.0, 0.0, 0.55}
\definecolor{darkpastelgreen}{rgb}{0.01, 0.75, 0.24}
\definecolor{darkgreen}{rgb}{0.0, 0.2, 0.13}
\definecolor{darkgoldenrod}{rgb}{0.72, 0.53, 0.04}
\definecolor{darkred}{rgb}{0.55, 0.0, 0.0}
\definecolor{forestgreen}{rgb}{0.13, 0.55, 0.13}
\hypersetup{
     colorlinks   = true,
     citecolor    = mahogany,
		 linkcolor		= forestgreen
}

\begin{document}
\title{Truly Perfect Samplers for Data Streams and Sliding Windows}
\author{
Rajesh Jayaram\thanks{Carnegie Mellon University and Google Research NYC. 
E-mail: \email{rkjayara@google.com}}\\
\and
David P. Woodruff\thanks{Carnegie Mellon University. 
E-mail: \email{dwoodruf@cs.cmu.edu}}\\
\and
Samson Zhou\thanks{Carnegie Mellon University. 
E-mail: \email{samsonzhou@gmail.com}}
}

\maketitle
\begin{abstract}
In the $G$-sampling problem, the goal is to output an index $i$ of a vector $f \in\mathbb{R}^n$, such that for all coordinates $j \in [n]$,
\[\textbf{Pr}[i=j] = (1 \pm \eps) \frac{G(f_j)}{\sum_{k\in[n]} G(f_k)} + \gamma,\]
where $G:\mathbb{R} \to \mathbb{R}_{\geq 0}$ is some non-negative function. If $\eps = 0$ and $\gamma = 1/\poly(n)$, the sampler is called \textit{perfect}. In the data stream model, $f$ is defined implicitly by a sequence of updates to its coordinates, and the goal is to design such a sampler in small space. Jayaram and Woodruff (FOCS 2018) gave the first perfect $L_p$ samplers in turnstile streams, where $G(x)=|x|^p$, using $\text{polylog}(n)$ space for $p\in(0,2]$. 
However, to date all known sampling algorithms are not \textit{truly perfect}, since their output distribution is only point-wise $\gamma = 1/\poly(n)$ close to the true distribution. This small error can be significant when samplers are run many times on successive portions of a stream, and leak potentially sensitive information about the data stream.
	 
In this work, we initiate the study of \textit{truly perfect} samplers, with $\eps = \gamma = 0$, and comprehensively investigate their complexity in the data stream and sliding window models. We begin by showing that sublinear space truly perfect sampling is impossible in the turnstile model, by proving a lower bound of $\Omega\left(\min\left\{n,\log \frac{1}{\gamma}\right\}\right)$ for any $G$-sampler with point-wise error $\gamma$ from the true distribution. We then give a general time-efficient sublinear-space framework for developing truly perfect samplers in the insertion-only streaming and sliding window models. As specific applications, our framework addresses $L_p$ sampling for all $p>0$, e.g., $\tO{n^{1-1/p}}$ space for $p\ge 1$, concave functions, and a large number of measure functions, including the $L_1-L_2$, Fair, Huber, and Tukey estimators. The update time of our truly perfect $L_p$-samplers is $\O{1}$, which is an exponential improvement over the running time of previous perfect $L_p$-samplers. 
\end{abstract}

\section{Introduction}
The streaming model of computation has emerged as an increasingly popular paradigm for analyzing massive data sets, such as network traffic monitoring logs, IoT sensor logs, financial market updates, e-commerce transaction logs, and scientific observations, as in computational biology, astronomy, or particle physics. 
In the (one-pass) streaming model, an underlying data set is implicitly defined through sequential updates that arrive one-by-one and can only be observed once, and the proposed algorithms are required to use space that is sublinear in the size of the input.

Sampling has proven to be a fundamental and flexible technique for the analysis of massive data.
A significant line of active work has studied sampling techniques~\cite{Vitter85, GemullaLH08, CohenCD11, CohenDKLT11, CohenCD12, CohenDKLT14, Haas16, Cohen18} in big data applications such as network traffic analysis~\cite{GilbertKMS01, EstanV03, MaiCSYZ06, HuangNGHJJT07, ThottanLJ10}, database analysis~\cite{LiptonNS90, HaasS92, LiptonN95, HaasNSS96, GibbonsM98, Haas16, CohenG19}, distributed computing~\cite{CormodeMYZ10, TirthapuraW11, CormodeMYZ12, WoodruffZ16, JayaramSTW19}, and data summarization~\cite{FriezeKV04, DeshpandeV06, DeshpandeRVW06, DeshpandeV07, AggarwalDK09, MahabadiRWZ20}. 
Given a non-negative weight function $G: \mathbb{R} \to \mathbb{R}_{\geq 0}$ and a vector $f \in \R^n$, the goal of a $G$-sampler is to return an index $i \in \{1,2,\dots , n\}$ with probability proportional to $G(f_i)$. In the data stream setting, the vector $f$, called the \emph{frequency vector}, is given by a sequence of $m$ updates (referred to as insertions or deletions) to its coordinates. More formally, in a data stream the vector $f$ is initialized to $0^n$, and then receives a stream of updates of the form $(i,\Delta) \in [n] \times \{-M,\dots,M\}$ for some integer $M > 0$. The update $(i,\Delta)$ causes the change $f_{i} \leftarrow f_{i} + \Delta$. This is known as the \textit{turnstile} model of streaming, as opposed to the \textit{insertion-only} model where $\Delta \geq 0$ for all updates. A $1$-pass $G$-sampler must return an index given only one pass through the updates of the stream.
 
The most well-studied weight functions are when $G(x) = |x|^p$ for some $p >0$. Such samplers in their generality, known as \textit{$L_p$ samplers}, were introduced by \cite{MonemizadehW10}. $L_p$ samplers have been used to develop efficient streaming algorithms for heavy hitters, $L_p$ estimation, cascaded norm approximation, and finding duplicates \cite{MonemizadehW10, AndoniKO11, JowhariST11, BravermanOZ12, JayaramW18,cohen2020wor}. 
Formally, a $G$-sampler is defined as follows. 
\begin{definition}[$G$ sampler]\label{def:lpsamp}
Let $f\in\mathbb{R}^n$, $0\le\eps,\gamma<1$, $0<\delta<1$ and $G: \mathbb{R} \to \mathbb{R}_{\geq 0}$ be a non-negative function satisfying $G(0) = 0$.  An $(\eps,\gamma,\delta)$-\emph{approximate $G$-sampler} is an algorithm that outputs an index $i\in[n]$ such that if $f \neq \vec{0}$, for each $j\in[n]$,
\begin{equation}\label{eqn:gsamp}
\PPr{i=j}=(1\pm\eps)\frac{G(f_j)}{\sum_{k \in [n]} G(f_k)} \pm \gamma
\end{equation}
and if $f =  \vec{0}$, then the algorithm outputs a symbol $\bot$ with probability at least $1-\delta$. 
%If $f \neq \vec{0}$, the sampler can only output $\bot $ with probability at most $\gamma$. 
The sampler is allowed to output $\fail$ with some probability $\delta$, in which case it returns nothing.

If $\eps=0$ and $\gamma = n^{-c}$, where $c>1$ is an arbitrarily large constant, then the sampler is called \emph{perfect}.
  If $\eps = 0$ and $\gamma=0$, then the sampler is called \textit{truly perfect}.
\end{definition}
In general, if $\eps>0$ and $\gamma = n^{-c}$, where $c>1$ is an arbitrarily large constant, an $(\eps,\gamma,\delta)$-sampler is commonly referred to as an $\eps$-relative error \textit{approximate sampler}.  Notice that the guarantees on the distribution of a sampler are all conditioned on the sampler not outputting $\fail$. In other words, conditioned on not outputting $\fail$, the sampler must output a value in $[n] \cup \{\bot\}$ from a distribution which satisfies the stated requirements. When $f = \vec{0}$, the distribution in equation \ref{eqn:gsamp} is undefined; therefore the special symbol $\bot$ is needed to indicate this possibility. 

In the case of $L_p$ samplers with $p > 0$, the underlying distribution is given by $|f_j|^p/\|f\|_p^p$. Such samplers are particularly useful as subroutines for other streaming and data-analytic tasks.  In the insertion-only model, the classical reservoir sampling technique of \cite{Vitter85} gives an $\O{\log n}$ space truly perfect $L_1$ sampling algorithm. However,  
when $p \neq 1$, or when negative updates are also allowed (i.e., the turnstile model), the problem becomes substantially more challenging. In fact, the question of whether such $L_p$ samplers even exist using sublinear space was posed by Cormode, Murthukrishnan, and Rozenbaum~\cite{CormodeMR05}. 

Monemizadeh and Woodruff partially answered this question by showing the existence of an $(\eps,n^{-c},1/2)$-approximate $L_p$ sampler for $p \in [1,2]$ using $\poly\left(\frac{c}{\eps},\log n\right)$ bits of space in the turnstile model~\cite{MonemizadehW10}.  
The space bounds were improved by Andoni, Krauthgamer, and Onak \cite{AndoniKO11}, and then later by Jowhari, Saglam, and Tardos \cite{JowhariST11}, to roughly $\O{c \eps^{-\max(1,p)}\log^2 n}$ for $p\in(0,2)$ and $\O{c \eps^{-2} \log^3 n}$ for $p=2$. This matched the lower bound of $\Omega(\log^2 n)$ in terms of $\log n$ factors for $p<2$, as shown in the same paper~\cite{JowhariST11}, but was loose in terms of $\eps$. This gap was explained by Jayaram and Woodruff~\cite{JayaramW18}, who gave the first \emph{perfect} $(0,n^{-c},1/2)$-$F_p$ samplers in the streaming model, using $\O{c \log^2 n}$ bits of space for $p\in(0,2)$ and $\O{c \log^3 n}$ bits of space for $p=2$.  For a further discussion on the development of $L_p$ samplers in the streaming model, we direct the reader to the survey~\cite{cormode2019p}. 
In addition to $L_p$ sampling, \cite{CohenG19} also considered samplers for certain classes of concave functions $G$ in the insertion-only model of streaming. 

\paragraph{Truly Perfect Sampling.}
Unfortunately, none of the aforementioned perfect samplers are \textit{truly perfect}.  Specifically, they have an additive error of $\gamma =n^{-c}$, and space depending linearly on $c$. While this may be acceptable for some purposes where only a small number of samples are required, this error can have significant downstream consequences when many samplers are run independently. For instance, a common usage of sampling for network monitoring and event detection is to generate samples on successive portions of the stream, which are reset periodically (e.g., minute by minute). Additionally, in a large, distributed database, many independent samplers can be run locally on disjoint portions of the dataset. These samples can be used as compact summaries of the database, providing informative statistics on the distribution of data across machines. While the samples generated on a single portion of the data may be accurate enough for that portion, the $1/\poly(n)$ variation distance between the samples and true distribution accumulates over many portions. For large databases containing $s$ distributed machines with $s\gg\poly(n)$, or for samplers run on short portions of high throughput streams, the resulting gap in variation distance between the joint distributions of the samples and the true distribution can blow up to a \textit{constant}. This results in large accuracy issues for sensitive tests run on the data. 

Moreover, this creates large issues for \textit{privacy}, even when the identities of samples are anonymized. 
For instance, a non-truly perfect sampler may positively bias a certain subset $S \subset [n]$ of coordinates when a given entry is in the dataset (i.e, $x_i \neq 0$), and may negatively bias $S$ if that entry is not present (i.e., $x_i = 0$). 
Given sufficiently many samples, an onlooker would be able to easily distinguish between these cases.
Our truly perfect samplers achieve perfect security~\cite{Data16}, whereas in previous work a streaming algorithm that has input $X$ could have an output that depends on arbitrary other information in $X$ and thus could in principle reveal information about every other entry in $X$. 
If the entries of $X$ correspond to sensitive data records, revealing as little as possible about $X$ is crucial, and is the basis for studying perfect security. 

Another important application of truly perfect sampling is in situations where samples from previous portions of the stream influence future portions of the stream, causing cascading blow-ups in error. For instance, samples and sketches can be used to approximate gradient updates for gradient descent~\cite{Johnson013,ZhaoZ15,NeedellSW16,ivkin2019communication}, where a large number of future gradients naturally depend on the samples generated from prior ones.  Unbiasedness is also important for interior point methods, since bias in estimates of the gradients can result in large drift, and therefore error, in the algorithm (see, e.g., Theorem 2 of ~\cite{hu2016bandit}).
Beyond non-adversarial adaptivity, we may also have a malicious attacker who uses adaptivity in an uncontrolled manner. For example, a malicious adversary can adaptively query a database for samples, with future queries depending on past samples. Such streams with adaptive updates are the focus of the field of \textit{adversarial robust streaming} \cite{MitrovicBNTC17,AvdiukhinMYZ19,ben2020adversarial,ben2020framework,hassidim2020adversarially,Woodruff20Tight,BravermanHMSSZ21,AttiasCSS21}. Due to this adaptivity, the variation distance between the joint distributions can increase \textit{exponentially}, causing large accuracy issues after only a small number of adaptive portions of the stream. Thus, even a perfect sampler would not avoid significant information leakage in such settings, and instead only a truly perfect sampler would be robust to drifts in the output distribution. Finally, truly perfect samplers are of fundamental importance in information-theoretic security. 
%Here, a common notion is that of a private approximation \cite{f06}, where a simulator given only the ability to generate samples, can simulate the distribution of samples output by the algorithm; a truly perfect sampler would thus be perfectly private, i.e., have no approximation in the simulation.  
%\Rajesh{Add applications of truly perfect sampling.}

%A unique characteristic of our samplers is that they are based on \textit{time-stamp} sampled updates, as opposed to the \textit{precision sampling} framework used to design approximate and perfect $L_p$ samplers \cite{MonemizadehW10,AndoniKO11,JowhariST11,JayaramW18}.
\paragraph{The Sliding Window Model.}
While applicable for many situations in data analysis, the standard streaming model does not capture situations in which the data is considered time-sensitive. In applications such as network monitoring~\cite{CormodeM05a,CormodeG08,Cormode13}, event detection in social media~\cite{OsborneEtAl2014}, and data summarization~\cite{ChenNZ16,EpastoLVZ17}, recent data is considered more accurate and important than data that arrived prior to a certain time window. To address such
settings, \cite{DatarGIM02} introduced the \emph{sliding window model}, where only the $W$ most recent updates to the stream induce the underlying input data, for some window size parameter $W>0$.
The most recent $W$ updates form the \emph{active data}, whereas updates previous to the $W$ most recent updates are \emph{expired}. 
The goal is to aggregate information about the active data using space sublinear in $W$. We remark that, generally speaking, the sliding window model is insertion-only by definition. 
Hence, the sliding window model is a strict generalization of the standard insertion-only streaming model.
 
The sliding window model is more appropriate than the unbounded streaming model in a number of applications~\cite{BabcockBDMW02, BravermanOZ12, MankuM12, PapapetrouGD15, WeiLLSDW16} and has been subsequently studied in a number of additional settings~\cite{LeeT06, LeeT06b, BravermanO07, DatarM07, BravermanOZ12, BravermanLLM15, BravermanLLM16, BravermanGLWZ18, BravermanDMMUWZ18}. To date, however, no truly perfect, perfect, or even approximate $L_p$ samplers for the sliding window model are known, leaving a substantive gap in our understanding of sampling for these models.

\subsection{Our Contributions}
In this work, we initiate the study of \textit{truly perfect samplers}, for general weight functions $G$ in the data stream and sliding window models. We begin by studying the problem of truly perfect sampling in the \textit{turnstile} model of streaming, where both positive and negative updates can be made to the coordinates in $f$. We demonstrate that in the turnstile model, the additive $1/\poly(n)$ error in previous approximate and perfect $L_p$ samplers is inherent \cite{MonemizadehW10,AndoniKO11,JowhariST11,JayaramW18}. 

\begin{restatable}{theorem}{thmtrulyperflb}
\thmlab{thm:trulyperflb}
Fix constant $\eps_0 < 1$, integer $r\ge 1$, and let $2^{-n/2} \leq \gamma < \frac{1}{4}$. Let $G : \R \to \R_{\geq 0}$ be any function satisfying $G(x) > 0 $ for $x \neq 0$, and $G(0) =0$. Then any $(\eps_0,\gamma,\frac{1}{2})$-approximate $G$-sampler $\mathcal{A}$ in the $r$-pass turnstile streaming model must use $\Omega\left(\min \left\{n,\log \frac{1}{\gamma}\right\}\right)$ bits of space.
\end{restatable}
\thmref{thm:trulyperflb} explains why all prior approximate and perfect samplers developed for the turnstile sliding window model paid a $1/\poly(n)$ additive error in their variation distance. In particular, when $\gamma = n^{-c}$, our lower bound of $\Omega(c \log n)$ for a $(\eps, \gamma, \frac{1}{2})$-$F_p$sampler for $p \in (0,2]$ is nearly tight, given the upper bound of $\O{c \log^2 n}$ of \cite{JayaramW18} for $p \in (0,2)$ and $\O{c \log^3 n}$ for $p=2$, which achieve perfect sampling ($\eps = 0)$. This demonstrates that $\log \frac{1}{\gamma}\polylog n$ is the correct complexity of $(0,\gamma,\frac{1}{2})$-$L_p$ sampling. 
Our lower bound is based on the fine-grained hardness of the $\equ$ problem from two-party communication complexity, demonstrating that a $(\eps,\gamma,1/2)$-$G$ sampler yields a communication protocol which solves $\equ$ with $\gamma$ advantage.

Given the strong impossibility results for designing truly perfect samplers in the turnstile model, we shift our attention to the fundamental insertion-only model. 
Given a measure function $G:\mathbb{R}\to\mathbb{R}^{\ge0}$ such that $G(x)=G(-x)$, $G(0)=0$ and $G$ is non-decreasing in $|x|$, we define $F_G=\sum_{i=1}^n G(f_i)$. 
We design a general framework for designing truly perfect $G$-samplers for a large number of useful functions $G$ in insertion-only streams and sliding windows with insertion-only updates. 
The framework is developed in Section \ref{sec:framework}, wherein several instantiations of the framework are given for specific functions $G$. 
Our theorem in its most general form is as follows, although we remark that for several applications, such as for $F_p$ estimation, significant additional work is needed to apply the theorem.

\begin{framework}
\framelab{frame:G:sampler}
Let $G$ be a function such that $0\le G(x)-G(x-1)\le\zeta$ for all $x\ge 1$. 
For insertion-only streams, there exists a perfect $G$ sampler that succeeds with probability at least $1-\delta$ and uses $\O{\frac{\zeta m}{F_G}\log n\log\frac{1}{\delta}}$ bits of space. 
For the sliding window model with insertion-only updates, there exists a truly perfect $G$ sampler that succeeds with probability at least $1-\delta$ and uses $\O{\frac{\zeta W}{F_G}\log^2 n\log\frac{1}{\delta}}$ bits of space. 
Further, the time to process each update is $\O{1}$ in expectation. 
(See \thmref{thm:mestimator:framework}.)
\end{framework}
%\begin{restatable}{theorem}{thmmestimatorframework}
%\thmlab{thm:mestimator:framework:intro}
%Let $G$ be a function such that $G(x)-G(x-1)\le\zeta$ for all $x\ge 1$. 
%Given a lower bound $\widehat{F_G}$ on $F_G$, then there exists a truly perfect $G$ sampler for an insertion-only stream of length $m$ that succeeds with probability at least $1-\delta$ and uses $\O{\frac{\zeta m}{\widehat{F_G}}\log n\log\frac{1}{\delta}}$ bits of space. Further, the time to process each update is $\O{1}$ in expectation. 
%\end{restatable}
The main barrier to applying \frameref{frame:G:sampler} to any arbitrary measure function $G$ is obtaining a ``good'' lower bound $\widehat{F_G}$ to $F_G=\sum_{i\in[n]}G(f_i)$. Moreover, this lower bound must be obtained correctly with probability $1$, as any possibility of failure of a randomized algorithm would necessarily contribute to additive error to the distribution of the samples, resulting in only a perfect, but not truly perfect, sampler. 

%On the other hand, there are no known lower bounds for $\eps$-approximate $L_p$ samplers, due to the existence of a perfect sampler~\cite{JayaramW18}. 
%On the other hand, by partitioning the stream into blocks of size $W$ and running reservoir sampling on each block, \cite{BravermanOZ12} provide a perfect $L_1$ sampler on sliding windows that only uses $\O{\log n}$ bits of space. 
%Unfortunately, the reservoir sampling technique of \cite{BravermanOZ12} does not seem to extend to perfect $L_p$ sampling on sliding windows for general $p$. 
%Nevertheless, due to this result, one might also guess that space complexity of $L_p$ sampling on sliding windows is actually $\O{\polylog(n)}$. 

Interestingly, our samplers utilize a timestamp-based reservoir sampling scheme, as opposed to the common \textit{precision sampling} framework used for other $L_p$ samplers~\cite{AndoniKO11,JowhariST11,JayaramW18,JayaramW18b,JayaramSTW19,chen2020improved}. This property of our samplers makes our framework particularly versatile, and for instance allows it to be applied to the sliding window model of streaming.   

As specific applications of our framework, we obtain the first truly perfect samplers for many fundamental sampling problems, including $L_p$ sampling, concave functions, and a large number of measure functions, including the $L_1-L_2$, Fair, Huber, and Tukey estimators. For $p \geq 1$, our results for $L_p$ sampling are as follows; we defer $p\in(0,1)$ to \thmref{thm:framework:general}. 

\begin{restatable}{theorem}{thmperfectlplarge}
\thmlab{thm:perfect:lp:large}
For the insertion-only streaming model and $p\ge 1$, there exists a truly perfect $L_p$ sampler that uses $\O{1}$ update time and $\O{n^{1-1/p}\,\polylog(n)}$ bits of space.
%sliding window model. 
%, and $\O{W^{1-1/p}\,\polylog(n)}$ bits in the sliding window model, where $W$ is the window size. 
\end{restatable}

Together, \thmref{thm:trulyperflb} and \thmref{thm:perfect:lp:large} show a strong separation between turnstile and insertion-only truly perfect $L_p$ samplers; surprisingly, for every $p > 1$, a truly perfect $L_p$ sampler exists with $\O{n^{1-1/p}\,\polylog(n)}$ space in the insertion-only model, while in the turnstile model this requires $\Omega(n)$ space. 

Another interesting feature of \thmref{thm:perfect:lp:large} is that the time to process each update is $\O{1}$! 
In contrast, the perfect samplers of \cite{JayaramW18}, which each are not truly perfect, have update time $n^{O(c)}$ to achieve variation distance $\frac{1}{n^c}$. Thus, we obtain an exponential improvement, and optimal running time. 

Yet another interesting feature of our algorithm in \thmref{thm:perfect:lp:large} is that it is sampling-based rather than sketching-based, and thus if the indices have metadata associated with them then we can additionally return that metadata, whereas the sketching-based algorithm of \cite{JayaramW18} cannot. For example, the indices may be keys we sample by and each key is part of some document; in this case we sample by the keys but additionally can return the document sampled. 

Notice also that the complexity in \thmref{thm:perfect:lp:large} above degrades as $p \to 1$. For instance, for $p=1$, our bound degrades to $n^{1-1/p}\log n=\log n$, which matches the complexity of reservoir sampling in the insertion only model, or the sliding window truly perfect $L_1$ sampler by \cite{BravermanOZ12}.

\paragraph{$M$-Estimators.} 
In addition to \thmref{thm:perfect:lp:large}, \frameref{frame:G:sampler} also implies truly perfect samplers for a number of $M$-estimators:
\begin{itemize}
\item
Truly perfect $G$ samplers for insertion-only streams that use $\O{\log n\log\frac{1}{\delta}}$ bits of space when $G$ is the $L_1-L_2$ estimator, the Fair estimator, the Huber estimator, or the Tukey estimator. 
(See \corref{cor:stream:mestimator} and \thmref{thm:stream:tukey}.)
\item
Truly perfect $G$ samplers for sliding windows with insertion-only updates that use $\O{\log^2 n\log\frac{1}{\delta}}$ bits of space when $G$ is the $L_1-L_2$ estimator, the Fair estimator, or the Huber estimator, or the Tukey estimator. 
(See \corref{cor:sw:mestimator} and \thmref{thm:sw:tukey}.)
\end{itemize}

\paragraph{Matrix Norms.} 
\frameref{frame:G:sampler} can also be extended to truly perfect sampling for matrix norms. 
That is, given a matrix $\M\in\mathbb{R}^{n\times d}$, the goal is to sample a row $\m_i$ of $\M$ with probability proportional to $G(\m_i)$ for some given function $G$. 
For example, when $G(\x)=\sqrt{\sum_{i\in[d]}x_i^2}$ is the $L_2$ norm of each row, then such a  row sampling primitive would be equivalent to $L_{1,2}$ sampling, which has recently been used in adaptive sampling techniques (see \cite{MahabadiRWZ20} and references therein). 
See \thmref{thm:matrix:framework} for more details. 

\paragraph{Turnstile Streams.} 
Next, we show that \frameref{frame:G:sampler} can also be extended to strict turnstile streams, which combined with \thmref{thm:trulyperflb} shows the separation of general and strict turnstile streams.
We give a general reduction as follows:
\begin{restatable}{theorem}{thmperfectmultipass}
\thmlab{thm:perfect:multi:pass}
Suppose there exists a truly perfect $L_p$ sampler in the one-pass insertion-only streaming model that uses $S$ bits of space. 
Then there exists a truly perfect $L_p$ sampler that uses $\tO{Sn^{\gamma}}$ space and $\O{\frac{1}{\gamma}}$ passes over a strict turnstile stream, which induces intermediate frequency vectors with nonnegative coordinates. 
\end{restatable}

\paragraph{Truly Perfect $F_0$ Sampling.}
We give a truly perfect sampler for the important $F_0$ problem for both the insertion-only streaming model and the sliding window model (see \thmref{thm:truly:perfect:F0} and \corref{cor:sw:truly:perfect:F0}). 
Our algorithm works by tracking the first $\sqrt{n}$ unique items in the stream to decide whether $F_0>\sqrt{n}$. 
If $F_0\le\sqrt{n}$, then it suffices to output a random item among those collected. 
Otherwise, we simultaneously generate a set $S$ of $\sqrt{n}$ random items, so that with constant probability, some item of $S$ appears in the stream. 
We can then output any item of $S$ that has appeared in the stream uniformly at random. 
Surprisingly, our algorithms use $\O{\sqrt{n}\log n}$ space whereas there exists a truly perfect $F_0$ sampler in the random oracle model with space $\O{\log n}$. 
We believe the complexity of truly perfect $F_0$ sampling without the assumption of a random oracle to be a fascinating open question.

\paragraph{Truly Perfect Sampling in the Random Order Model.} 
Next, in \appref{app:randomorder}, we demonstrate that for \textit{random order} streams, we can design a truly perfect $L_2$ sampling using only $\O{\log^2 n}$ bits of space. Since, as in our framework for truly perfect sampling, our algorithms are timestamp based, they also apply to the more challenging sliding window model of streaming. The complexity of our sampler is a $\log n$ factor smaller than the complexity of the best known previous samplers \cite{JayaramW18} in the adversarial order model, which had $\gamma = 1/\poly(n)$ additive error in their distribution, and which did not apply to the sliding window model.
Our theorem for $L_2$ sampling on random order streams is as follows.

\begin{restatable}{theorem}{thmsmallp}
\thmlab{thm:small:p}
There exists a one-pass sliding window algorithm for random order insertion-only streams that outputs index $i\in[n]$ with probability $\frac{f_i^2}{F_2}$ and outputs $\fail$ with probability at most $\frac{1}{3}$, i.e., the algorithm is a truly perfect $L_2$ sampler, using $\O{\log^2 n}$ bits of space and $\O{1}$ update time. 
\end{restatable}

We generalize this approach to perfect $L_p$ samplers on random order streams for integers $p>2$. 
%\begin{restatable}{theorem}{thmlargep}
%\thmlab{thm:large:p}
%Let $p>2$ be a fixed integer.
%There exists a one-pass sliding window algorithm for random order insertion only streams that outputs index $i\in[n]$ with probability $\frac{f_i^p}{\sum_{j=1}^n f_j^p}$, and outputs $\fail$ with probability at most $\frac{1}{3}$, i.e., the algorithm is a truly perfect $L_p$ sampler, using $\O{W^{1-\frac{1}{p-1}}\log n}$ bits of space and $\O{1}$ update time.
%% and $\O{\log n}$ amortized update time. 
%\end{restatable}
\begin{restatable}{theorem}{thmlargep}
\thmlab{thm:large:p}
Let $p>2$ be a fixed integer.
There exists a one-pass algorithm that outputs an index $i\in[n]$ with probability $\frac{f_i^p}{\sum_{j=1}^n f_j^p}$, and outputs $\fail$ with probability at most $\frac{1}{3}$ on a random-order insertion-only stream of length $m$ , i.e., the algorithm is a truly perfect $L_p$ sampler, using $\O{m^{1-\frac{1}{p-1}}\log n}$ bits of space and $\O{1}$ update time.
% and $\O{\log n}$ amortized update time. 
\end{restatable}

For $p=2$, intuitively our algorithm follows by tracking collisions between adjacent elements in the stream. Here, a collision occurs when to subsequent updates are made to the same coordinate $i \in [n]$. The probability that this occurs at a given timestep is $\frac{f_i(f_i-1)}{m(m-1)}$. Since this is not quite the right probability, we ``correct'' this distribution by a two part rejection sampling step to obtain a truly perfect sampler. For truly perfect $L_p$ sampling on random order streams for integers $p>2$, we store consecutive blocks of $m^{1-\frac{1}{p-1}}$ elements in the stream, along with their corresponding timestamps, and instead look for $p$-wise collisions within the block.

\paragraph{Fast Perfect $L_p$ Sampling for $p<1$.} 
Lastly, we demonstrate that for insertion only streams, the runtime of the perfect $L_p$ samplers of \cite{JayaramW18} can be significantly improved. Specifically, these prior samplers had update time which was a large polynomial in $n$: roughly, to obtain a $(0,\gamma,1/2)$ $L_p$-sampler, these algorithms required $\O{\frac{1}{\gamma}}$ runtime. This poses a serious barrier for usage of these perfect samplers in most applications. We demonstrate, however, that the tradeoff between runtime and distributional error in a perfect $L_p$ sampler is not required for $p<1$, by giving an algorithm with nearly tight space complexity, which achieves $\poly(\log n)$ update time. 
See \corref{cor:perfect:smallp:fast} for more details. 

\subsection{Our Techniques}
\paragraph{Truly Perfect $L_p$ Samplers.} 
We begin by describing our sampling framework for the case of $L_p$ samplers. Specifically, to obtain a truly perfect $L_p$ sampler for insertion-only streams of length $m$ and for $p \geq 1$, we run 
$\O{n^{1-1/p}}$ parallel instances of a single sampler, which uses only $\log n$ bits of space, but succeeds with probability $\Omega\left(\frac{1}{n^{1-1/p}}\right)$. 
Each instance of the single sampler first applies reservoir sampling to the updates in the stream to sample an item $s \in [n]$, along with a specific timestamp $t_s$ when $s$ was added to the reservoir sample, and keeps a counter $c$ of how many times $s$ appears in the stream after it is first sampled. 
 
Inspired by a technique of Alon, Matias, and Szegedy for $F_p$-estimation \cite{AMS99}, we then prove that if $c$ occurrences of $s$ appear afterwards and we output $s$ with probability proportional to $c^p-(c-1)^p$, then by a telescoping argument, the probability of outputting each $i\in[n]$ is proportional to $|f_i|^p$.  
Thus a sampler that successfully outputs a coordinate must do so from the desired distribution. To implement the rejection sampling step, we need to obtain a good normalizing factor so that the resulting step forms a valid distribution. We demonstrate that it suffices to estimate $\|f\|_\infty$ to obtain a good normalizing factor, which results in acceptance probability of at least $\Omega\left(\frac{1}{n^{1-1/p}}\right)$. We carry out this step deterministically with the Misra-Gries sketch, which is necessary since failure of any randomized estimation algorithm would introduce additive error into the sampler. 
By repeating $\O{n^{1-1/p}}$ times, we ensure that at least one sampler succeeds with constant probability, showing that we output $\fail$ with at most constant probability. 
We use a similar approach for $p\in(0,1)$. 

\paragraph{General Framework for Truly Perfect Sampler.}
The aforementioned telescoping argument for our truly perfect sampler for insertion-only updates can be viewed as using the identity $\sum_{c=1}^{f_i}(G(c)-G(c-1))=G(f_i)-G(0)$ for $G(x)=x^p$. 
By the same reasoning, we can generalize the paradigm of ``correcting'' the sampling probability for any monotonic function $G$ with $G(0)=0$ by a subsequent rejection sampling step though the size of the acceptance probability depends on obtaining a good normalizing factor. 
%%The reason we needed to repeat $\O{W^{1-1/p}}$ times in the $L_p$ sampling case was that each instance only succeeded with probability roughly $\frac{1}{W^{1-1/p}}$. 
For $F_G=\sum_{i=1}^n G(f_i)$, we show that we can similarly bound the probability each single sampler succeeds with probability roughly $\frac{F_G}{m}$. 
Thus we require roughly $\frac{m}{F_G}$ parallel instances of the single sampler. 

For the case of the sliding window model, we can automatically expire the sample $(s,t_s)$ as soon as $t_s$ leaves the active window, causing an additional complexity in the sliding window model. 
However, by maintaining a number of ``checkpoints'', we can ensure that $(s,t_s)$ is an active element with constant probability.

\paragraph{Random-Order Sampling.} 
Finally, we improve our perfect $L_p$ samplers for random-order insertion-only streams by using distributional properties of each stream update. 
Namely, we modify our timestamp based sampling scheme to randomly sample a number of $p$-tuples and search for collisions among the $p$-tuples. 
For $p=2$, the idea is to consider two adjacent elements and see if they collide. 
Intuitively, an arbitrary position in the window is item $i$ with probability $\frac{f_i}{m}$ due to the random order of the stream, where $m$ is the length of the stream. 
The next update in the stream is also item $i$ with probability $\frac{f_i-1}{m-1}$. 
Thus the probability that both the two positions are updates to coordinate $i$ is $\frac{f_i(f_i-1)}{m(m-1)}$, which is not quite the right probability. 
Instead of using a telescoping argument as in the general framework, we instead ``correct'' this probability by sampling item $i$ in a position with probability $\frac{1}{m}$. 
Otherwise, with probability $1-\frac{1}{m}$, we sample item $i$ if the item is in the next position as well. 
Now the probability of sampling $i$ on the two adjacent elements is $\frac{1}{m}\frac{f_i}{m}+\frac{m-1}{m}\frac{f_i}{m}\frac{f_i-1}{m-1}=\frac{f_i^2}{m^2}$. 
We similarly generalize this argument to integer $p>2$. 

\subsection{Preliminaries}
We use $\mathbb{R}^{\ge0}$ to denote a non-negative number. 
We use the notation $[n]$ to represent the set $\{1,\ldots,n\}$ for any positive integer $n$. 
We use $\poly(n)$ to denote a fixed constant degree polynomial in $n$ but we write $\frac{1}{\poly(n)}$ to denote an arbitrary degree polynomial in $n$ that can be determined from setting constants appropriately.  
When an event has probability $1-\frac{1}{\poly(n)}$ of occurring, we say the event occurs with high probability. 
We similarly use $\polylog(n)$ to omit terms that are polynomial in $\log n$.  We note that all space bounds in this paper are given in \textit{bits}. 
In the insertion-only model of streaming, there is a vector $f \in \R^{n}$ which is initialized to the $0$ vector. The stream then proceeds with a sequence of $m = \poly(n)$ updates $i_1,i_2,\dots,i_m$ (the exact value of $m$ is unknown to the algorithm). After the first $t$ time steps, the state of the vector, denoted  $f^{(t)},$ is given by $f^{(t)}_j = \sum_{t' \leq t} \mathbf{1}(i_{t'} = j)$. 
%In the $G$-sampling problem, the algorithm must output its sample at the end of the stream. 
In the sliding window model, at time step $t$ only the most recent $W$ updates form the active portion of the stream, so that in the sliding window model: $f^{(t)}_j\sum_{t'\in(t-W,t]} \mathbf{1}(i_{t'} = j)$.
%For an underlying vector $f\in\mathbb{R}^n$, we use the notation $F_p=\sum_{i=1}^n f_i^p$. 

\section{Lower Bound for Truly Perfect Sampling in Turnstile Streams}\label{sec:lb}
In this section, we demonstrate that truly perfect $G$ samplers cannot exist in sublinear space in the turnstile model. Specifically, we show that any perfect sampler with additive error $\gamma=n^{-c}$ requires space at least $\Omega(c \log n)$. This demonstrates that no sublinear space truly perfect sampler can exist in the turnstile model, and demonstrates the tightness (up to $\log n$ factors), of the previously known perfect and approximate $L_p$ samplers \cite{MonemizadehW10,AndoniKO11,JowhariST11,JayaramW18}.

Our lower bound is based on the fine-grained hardness of the $\equ$ problem from two-party communication complexity \cite{brody2014certifying}. Specifically, consider the boolean function $\eq_n: \{0,1\}^n \times \{0,1\}^n \to \{0,1\}$ given by $\eq_n(x,y) = 1 \iff x =y$. In the two party, one way communication model, there are two parties: Alice and Bob. Alice is given a input string $x \in \{0,1\}^n$ and Bob is given $y \in \{0,1\}^n$. Then Alice must send a single message $M$ to Bob, who must then output whether $\eq_n(x,y)$ correctly with some probability. A communication protocol $\mathcal{P}$ is a randomized two-party algorithm which takes the input $(x,y)$ and decides on a message $M$ and output procedure $\out$ for Bob given $(M,y)$. The communication cost of $\mathcal{P}$ is denoted $\cost(\mathcal{P},x,y)$, and defined as the maximum number of bits sent in the message $M$ over all coin flips of the algorithm, on inputs $(x,y)$.
We now define the randomized \textit{refutation complexity} of a communication protocol for computing a Boolean function $f$. We define the \textit{refutation cost}, \textit{refutation error}, and \textit{verification error} as:
\begin{equation*}
\begin{split}
\rcost(\cP)& = \max_{(x,y) \in f^{-1}(0)} \cost(\cP,x,y)  \\
\rerr(\cP) &=  \max_{(x,y) \in f^{-1}(0)} \PPr{\out(\cP(x,y)) = 1} \\
\verr(\cP) &=  \max_{(x,y) \in f^{-1}(1)} \PPr{\out(\cP(x,y)) = 0}.
\end{split}
\end{equation*}
We define the randomized refutation complexity of a function $f$ for an integer $r\ge 1$ as 
\[R^{(r),\reff}_{\eps,\delta}(f) = \min_\cP \left\{\rcost(\cP) : 	\rerr(\cP) \leq \eps, \verr(\cP) \leq \delta 	\right\} \]
 where the minimum is restricted to $r$-round communication protocols $\cP$. 
Observe that the ``trivial'' one-round protocol for $\eq_n$ achieves $\eps$ refutation error and communicates $\min(n,\log(1/\eps))$ bits, so it is as though the instance size drops from $n$ to $\min(n,\log(1/\eps))$ when $\eps$ refutation error is allowed.
Thus we define the effective instance size as 
 \[	\hat{n} = \min\left\{n + \log(1-\delta), \; \log\left(\frac{\left(1-\delta\right)^2}{\eps}\right)\right\}.	\]
\begin{theorem}[Theorem 44  \cite{brody2014certifying}]\thmlab{thm:BCKlb}
	We have $R^{(r),\reff}_{\eps,\delta}(\eq_n) \geq \frac{1}{8}(1-\delta)^2(\hat{n} + \log(1-\delta) -5)$.
\end{theorem}
\noindent
We show \thmref{thm:trulyperflb} by giving a reduction from $\equ$ and applying \thmref{thm:BCKlb}:
\thmtrulyperflb*
\begin{proof}
Given such a sampler $\mathcal{A}$, we give an algorithm for the two-party, $r$-round equality problem as follows. 
Alice is given $x \in \{0,1\}^n$, and creates a stream with frequency vector given by $f = x$. 
Bob then adds the vector $-y$ into the stream so that the final state of the frequency vector induced by the stream is $f=x-y$. 
Alice and Bob each run $\mathcal{A}$ on their stream and repeatedly pass the state of the algorithm between each other over the course of $r$ rounds. 
Bob then finally obtains the output of the streaming algorithm $\mathcal{A}(f)$ after $r$ rounds. If the output is $\fail$, or anything except for the symbol $\bot$, then Bob declares $\eq_n(x,y) =0$. If the output is $\bot$,  Bob declares $\eq_n(x,y) =1$. Notice by definition of a $(\eps_0,\gamma,\frac{1}{2})$-$G$ sampler (Definition \ref{def:lpsamp}), if we actually had $x=y$, then $f = \vec{0}$, so if $\mathcal{A}$ does not output $\fail$, then it must declare $\bot$ with probability at least $1-\gamma$. Moreover, if $x \neq y$, then since $G((x-y)_i) > 0$ for some $i$, a correct sampler can output $\bot$ with probability at most $\gamma$. Furthermore, it can output $\fail$ in both cases with probability at most $\frac{1}{2}$. 
	
	The above protocol therefore satisfies that if $\eq_n(x,y) =0$, Bob outputs $1$ with probability at most $\gamma$, thus the refutation error is at most $\eps < \gamma$. Moreover, if $\eq_n(x,y) =1$, then $\mathcal{A}$ outputs fail with probability $\frac{1}{2}$, and conditioned on not outputting fail it must output $\bot$ with probability at least $1-\gamma$. Thus, the verification error is at most $\delta < 1/2 + \gamma < 3/4$. Then we have $n - \log(1-\delta) > n/2$, and $\log(\frac{(1-\delta)^2}{\eps}) > \log(\frac{1}{16 \gamma})$. Thus the effective input size is given by 	
	\[	\hat{n} > \min \left\{\frac{n}{2}, \log \frac{1}{16\gamma}\right\} \]
	Thus, by \thmref{thm:BCKlb}, we have 
	\begin{equation}
	\begin{split}
	R^{(r),\reff}_{\eps,\delta}(\eq_n) &\geq \frac{1}{8}(1-\delta)^2(\hat{n} + \log(1-\delta) -5) \\
	&\geq \frac{1}{8\cdot 16}(\hat{n} - 7) \\ 
	& = \Omega(\hat{n}) \\
	\end{split}
	\end{equation}
	which completes the proof of the lower bound. 
\end{proof}
\section{Framework for Truly Perfect Sampling}\label{sec:framework}
In this section, we first give a framework for truly perfect sampling for some measure function $G:\mathbb{R}\to\mathbb{R}^{\ge0}$ such that $G(x)=G(-x)$, $G(0)=0$ and $G$ is non-decreasing in $|x|$. 
If we define $F_G=\sum_{i=1}^n G(f_i)$, then we say that a truly perfect $G$ sampler outputs index $i\in[n]$ with probability $\frac{G(f_i)}{F_G}$. 
We then show how to apply the framework to $L_p$ sampling where $G(x)=|x|^p$ and to various $M$-estimators, such as the $L_1-L_2$, Fair, Huber, and Tukey estimators. 

\subsection{Algorithmic Framework}
Our algorithm is based on running parallel instances of a single sampler.
Each instance uses $\log n$ bits of space, but only succeeds with small probability and thus we need to run multiple instances to ensure that with sufficiently high probability, some instance succeeds.  
Each instance uses reservoir sampling to sample an item $s$ and keeps a counter $c$ of how many times $s$ appears in the stream after it is sampled. 

We first describe the $\sampler$ algorithm. 
Given a stream of elements $u_1,\ldots,u_m$, where each $u_i\in[n]$, $\sampler$ selects an index $j\in[m]$ uniformly at random and outputs $u_j$ as well as the number of instances of $u_j$ that appear after time $j$. 
The algorithm uses reservoir sampling to ensure that each item is selected with probability $\frac{1}{m}$. 
A counter is also maintained to track the number of instances of the sample. 
Each time a new sample replaces the existing sample in the reservoir sampling procedure, the counter is reset to zero. 

\begin{algorithm}[!htb]
\caption{\sampler: Reservoir sampling, counting number of times item has appeared afterwards.}
\alglab{alg:sampler}
\begin{algorithmic}[1]
\Require{A stream of updates $u_1,u_2,\ldots,u_m$, where each $u_i\in[n]$ represents a single update to a coordinate of the underlying vector $f$.}
\Ensure{Sample each coordinate $u_i$ with probability $\frac{1}{m}$ and output the number of occurrences that appears afterwards.}
\State{$s\gets\emptyset$, $c\gets 0$}
\For{each update $u_r$}
\State{$s\gets u_r$ with probability $\frac{1}{r}$}
\If{$s$ is updated to $u_r$}
\State{$c\gets 0$}
\Comment{Reset counter.}
\EndIf
\If{$u_r=s$}
\State{$c\gets c+1$}
\Comment{Increment counter.}
\EndIf
\EndFor
\State{\Return $s$ and $c$.}
\end{algorithmic}
\end{algorithm}

By outputting $s$ with probability $\frac{G(c)-G(c-1)}{\zeta}$, where $\zeta$ is a parameter such that $G(x)-G(x-1)\le\zeta$ for all possible coordinates $x$ in the frequency vector, i.e., $x\in\{f_1,\ldots,f_n\}$, then it can be shown by a telescoping argument that the probability of outputting each $i\in[n]$ is ``corrected'' to roughly $\frac{G(f_i)}{\zeta\cdot m}$, where $m$ is the length of the stream.  
Hence if the sampler successfully outputs a coordinate, it follows the desired distribution. 

\begin{algorithm}[!htb]
\caption{Truly perfect $G$-sampler algorithm for insertion only streams.}
\alglab{alg:perfect:framework}
\begin{algorithmic}[1]
\Require{A stream of updates $u_1,u_2,\ldots,u_m$, where each $u_i\in[n]$ represents a single update to a coordinate of the underlying vector $f$, a measure function $G$.}
\State{Initialize an instance of $\sampler$.}
\Comment{\algref{alg:sampler}}
\For{each update $u_t\in[n]$}
\State{Update $\sampler$ with $u_t$.}
\EndFor
\State{Let $s$ be the sampled output of $\sampler$ and let $c$ be the number of times $s$ has appeared afterwards.}
\State{Let $\zeta$ be a parameter such that $G(x)-G(x-1)\le\zeta$ for all $x\ge 1$.}
\State{\Return $s$ with probability $\frac{G(c+1)-G(c)}{\zeta}$.}
%\State{Otherwise, \Return $\bot$.}
\end{algorithmic}
\end{algorithm}

\begin{theorem}
\thmlab{thm:mestimator:framework}
Let $G$ be a function such that $0\le G(x)-G(x-1)\le\zeta$ for all $x\ge 1$. 
Given a lower bound $\widehat{F_G}$ on $F_G=\sum_{i\in[n]}G(f_i)$, then there exists a truly perfect $G$ sampler for an insertion-only stream that outputs $\fail$ with probability at most $\delta$ and uses $\O{\frac{\zeta m}{\widehat{F_G}}\log n\log\frac{1}{\delta}}$ bits of space. 
Further, the time to process each update is $\O{1}$ in expectation. 
\end{theorem}
\begin{proof}
	The probability that $s$ is the $j\th$ particular instance of item $i$ inside the stream is $\frac{1}{m}$. 
	Since the number of instances of $i$ appearing after $j$ is $f_i-j$ then the probability that $i$ is output is 
	\[\sum_{j=1}^{f_i}\frac{1}{m}\frac{G(f_i-j+1)-G(f_i-j)}{\zeta}=\frac{G(f_i)}{\zeta m}.\]
	We note that $G(f_i-j+1)-G(f_i-j)\le\zeta$ for all $j\in[f_i]$, so returning $s$ with probability $\frac{G(c+1)-G(c)}{\zeta}$ is valid.
	
	Hence the probability that some index is returned by \algref{alg:perfect:framework} is $\sum_{i\in[n]}\frac{G(f_i)}{\zeta m}=\frac{F_G}{\zeta m}$, where $F_G=\sum G(f_i)$. 
	Thus by repeating the sampler $\O{\frac{\zeta m}{F_G}\log\frac{1}{\delta}}$ times, the algorithm will output a sample $s$ with probability at least $1-\delta$. 
	Although the algorithm does not actually have the value of $F_G$, given a lower bound $\widehat{F_G}$ on $F_G$, then it suffices to repeat the sampler $\O{\frac{\zeta m}{\widehat{F_G}}\log\frac{1}{\delta}}$ times. 
	Moreover, the sample $s$ will output each index $i\in[n]$ with probability $\frac{G(f_i)}{F_G}$. 
	Each instance only requires $\O{\log n}$ bits to maintain the counter $c$, assuming $\log m=\O{\log n}$. 
	Thus the total space used is $\O{\frac{\zeta m}{\widehat{F_G}}\log n\log\frac{1}{\delta}}$ bits of space.
	
	We remark that the runtime of the algorithm can be optimized to constant time per update by storing a hash table containing a count and a list of offsets. Specifically, when item $i$ is first sampled by some repetition of the algorithm, then we start counting the number of subsequent instances of $i$ in the stream. If $i$ is subsequently sampled by another independent instance of the reservoir sampling at some time $t$, then it suffices to store the value of the counter at the time $t$ as an offset. This value does not change and can now be used to correctly recover the correct count of the number of instances of $i$ after time $t$ by subtracting this offset from the largest count. We can maintain a hash table with pointers to the head and tail of the list, so that when an item $i$ is sampled, we can efficiently check whether that item is already being tracked by another repetition of the sampler. Hence the update time is $\O{1}$ worst-case once the hash bucket for $i$ is determined and $\O{1}$ in expectation overall given the assignment of the bucket by the hash function. Note that by design of the offsets, we can build the correct counters at the end of the stream to determine the corresponding sampling probabilities. Finally, we observe that it is well-known how to optimize the runtime of the reservoir sampling over \algref{alg:sampler} by using $\O{k\log n}$ total time to sample $k$ items~\cite{li1994reservoir}. 
\end{proof}

An interesting corollary of our results is to obtaining $s=\tilde{o}(n^{1/p})$ samples, rather than only a single sample. 
Our memory, like previous (non-truly perfect) samplers, does get multiplied by $s$, but our update time surprisingly remains $\O{1}$. 
By comparison, the only known previous perfect $L_p$ sampler would require a prohibitive $s\cdot\poly(n)$ update time~\cite{JayaramW18}. 
The reason our algorithm does not suffer from a multiplicative overhead in update time is because our data structure maintains a hash table containing a count and a list of offsets for each of the distinct items that is sampled during the stream. 
Thus, if our algorithm is required to output $s$ samples, then our hash table now contains more items, but each stream update can only affect the counter and offset corresponding to the value of the update.

\subsection{Applications in Data Streams}
\seclab{sec:apps:stream}
In this section, we show various applications of \algref{alg:perfect:framework} in the streaming model. 
The main barrier to applying \thmref{thm:mestimator:framework} to any arbitrary measure function $G$ is obtaining a ``good'' lower bound $\widehat{F_G}$ to $F_G=\sum_{i\in[n]}G(f_i)$. 

\subsubsection{Truly Perfect $L_p$ Sampling on Insertion-Only Streams}
We first consider truly perfect $L_p$ sampling, where $G(x)=|x|^p$, for $p\ge 1$. 
Note that reservoir sampling is already a perfect $L_1$ sampler for $p=1$ and it uses $\O{\log n}$ bits of space on a stream of length $m=\poly(n)$. 
For $p\in(1,2]$, we first require the following norm estimation algorithm on insertion-only streams. 

We now introduce an algorithm that for \textit{truly perfect }$L_p$ sampling using the above framework. 
We first describe the case of $p=2$ but before describing our perfect $L_2$ sampler, we first recall the $\misragries$ data structure for finding heavy-hitters. 
\begin{theorem}
\cite{MisraG82}
\thmlab{thm:misragries}
There exists a deterministic one-pass streaming algorithm $\misragries$ that uses $\O{\frac{1}{\eps}\log m}$ space on a stream of length $m$ and outputs a list $L$ of size $\frac{1}{2\eps}$ that includes all items $i$ such that $f_i>2\eps m$. 
Moreover, the algorithm returns an estimate $\widehat{f_i}$ for each $i\in L$ such that $f_i-\eps m\le\widehat{f_i}\le f_i$.  
\end{theorem}

Although we could obtain a perfect $L_p$ sampler using any $L_p$ estimation algorithm that succeeds with high probability, we can further remove the additive $\frac{1}{\poly(n)}$ error of returning each coordinate $i\in[n]$ using the $\misragries$ data structure to obtain a \emph{truly} perfect $L_p$ sampler.

\begin{restatable}{theorem}{thmframeworkfour}
\thmlab{thm:framework:general}
For a frequency vector $f$ implicitly defined by an insertion-only stream, there exists an algorithm that returns each $i\in[n]$ with probability $\frac{f_i^p}{\sum_{j\in[n]} f_j^p}$, using space $\O{m^{1-p}\log n}$ for $p\in(0,1]$ and $\O{n^{1-1/p}\log n}$ bits of space for $p\in[1,2]$. 
\end{restatable}
We break down the proof of \thmref{thm:framework:general} into casework for $p\in[1,2]$ and $p\in(0,1]$. 
\begin{restatable}{theorem}{thmperfectmedp}
\thmlab{thm:framework3}
For $p\in[1,2]$ and a frequency vector $f$ implicitly defined by an insertion-only stream, there exists an algorithm that returns each $i\in[n]$ with probability $\frac{f_i^p}{\sum_{j\in[n]} f_j^p}$, using $\O{n^{1-1/p}\log n}$ bits of space. 
\end{restatable}
\begin{proof}
	By \thmref{thm:misragries}, using a single $\misragries$ data structure with $\O{n^{1-1/p}\log n}$ bits of space allows us to obtain a number $Z$ such that 
	\[\|f\|_{\infty}\le Z\le\|f\|_{\infty}+\frac{m}{n^{1-1/p}}.\]
	Note that for $p\in[1,2]$, the function $x^p-(x-1)^p$ is maximized at $p=2$, equal to $2x^{p-1}$, by the generalized binomial theorem. 
	Since $f_i\le\|f\|_\infty$, then we have $(f_i)^p-(f_i-1)^p\le 2\|f\|^{p-1}_{\infty}$ for any $i\in[n]$, so that $\zeta=2Z^{p-1}$ induces a valid sampling procedure. 
	Hence each instance outputs some index $i\in[n]$ with probability at least $\frac{F_p}{2Z^{p-1}m}$. 
	If $\|f\|_\infty\ge\frac{m}{n^{1-1/p}}$, then we have $2Z\le4\|f\|_\infty\le4\|f\|_p$, so that 
	\[\frac{F_p}{2Z^{p-1}m}\ge\frac{F_p}{4L_p^{p-1}\cdot m}=\frac{L_p}{4F_1}\ge\frac{1}{4n^{1-1/p}}.\]
	On the other hand, if $\|f\|_\infty\le\frac{m}{n^{1-1/p}}$, then we have $2Z\le\frac{4m}{n^{1-1/p}}$, so that
	\[\frac{F_p}{2Z^{p-1}m}\ge\frac{F_p\cdot n^{(p-1)^2/p}}{4m^p}=\frac{F_p\cdot n^{(p-1)^2/p}}{4F_1^p}\ge\frac{F_p\cdot n^{(p-1)^2/p}}{4F_p\cdot n^{(p-1)}}\]\[=\frac{1}{4n^{1-1/p}}.\]
	Therefore, the probability that an instance outputs some index $i\in[n]$ is at least $\frac{1}{4n^{1-1/p}}$, and it suffices to use $\O{n^{1-1/p}}$ such instances, with total space $\O{n^{1-1/p}\log n}$ bits of space. 
	By \thmref{thm:misragries}, conditioned on an index being returned by the algorithm, the probability that each coordinate $i\in[n]$ is output is $\frac{f_i^p}{F_p}$. 
\end{proof}

\begin{restatable}{theorem}{thmperfectlowp}
\thmlab{thm:framework4}
For $p\in(0,1]$ and a frequency vector $f$ implicitly defined by an insertion-only stream, there exists an algorithm that returns each $i\in[n]$ with probability $\frac{f_i^p}{\sum_{j\in[n]} f_j^p}$, using $\O{m^{1-p}\log n}$ bits of space. 
\end{restatable}
\begin{proof}
	Note that for $p\in(0,1]$, we have $(f_i)^p-(f_i-1)^p\le 1$ for any $i\in[n]$, so that $\zeta=1$ induces a valid sampling procedure. 
	Hence each instance outputs some index $i\in[n]$ with probability at least $\frac{F_p}{m}\ge\frac{1}{m^{1-p}}$. 
	Therefore, it suffices to use $\O{m^{1-p}}$ such instances, with total space $\O{m^{1-p}\log n}$ bits of space and conditioned on an index being returned by the algorithm, the probability that each coordinate $i\in[n]$ is output is $\frac{f_i^p}{F_p}$. 
\end{proof}
Together, \thmref{thm:framework3} and \thmref{thm:framework4} give \thmref{thm:framework:general}. 

\subsubsection{$M$-estimators on Insertion-Only Streams}
\seclab{sec:mestimators}
We generalize the paradigm of \algref{alg:perfect:framework} to sampling from general statistical $M$-estimator distributions. 
Recall that for a given a measure function $G:\mathbb{R}\to\mathbb{R}^{\ge0}$ such that $G(x)=G(-x)$, $G(0)=0$ and $G$ is non-decreasing in $|x|$, we define $F_G=\sum_{i=1}^n G(f_i)$. 
Then a truly perfect $M$-estimator sampler outputs index $i\in[n]$ with probability $\frac{G(f_i)}{F_G}$. 

For the $L_1-L_2$ estimator, we have $G(x)=2\left(\sqrt{1+\frac{x^2}{2}}-1\right)$ so that $G(x)-G(x-1)<3$ for $x\ge 1$. 
For the Fair estimator, we have $G(x)=\tau|x|-\tau^2\log\left(1+\frac{|x|}{\tau}\right)$ for some constant $\tau>0$ so that $G(x)-G(x-1)<\tau$ for $x\ge 1$. 
For the Huber measure function, we have $G(x)=\frac{x^2}{2\tau}$ for $|x|\le \tau$ and $G(x)=|x|-\frac{\tau}{2}$ otherwise, where $\tau>0$ is some constant parameter. 

\begin{restatable}{corollary}{corstreammestimator}
\corlab{cor:stream:mestimator}
There exist truly perfect $G$ samplers for the insertion-only streaming model that succeed with probability at least $1-\delta$ and use $\O{\log n\log\frac{1}{\delta}}$ bits of space when $G$ is the $L_1-L_2$ estimator, the Fair estimator, or the Huber estimator. 
\end{restatable}
\begin{proof}
For the $L_1-L_2$ estimator, we have $G(x)=2\left(\sqrt{1+\frac{x^2}{2}}-1\right)$ so that $G(x)-G(x-1)<3$ for $x\ge 1$. 
Moreover, $G(x)>|x|$ so $F_G>m$. 
Hence by \thmref{thm:mestimator:framework}, there exists a perfect $G$ sampler that uses $\O{\log n}$ bits of space when $G$ is the $L_1-L_2$ estimator.
	
For the Fair estimator, we have $G(x)=\tau|x|-\tau^2\log\left(1+\frac{|x|}{\tau}\right)$ for some constant $\tau>0$ so that $G(x)-G(x-1)<\tau$ for $x\ge 1$. 
Since $G(x)>\tau|x|$ and thus $F_G>\tau m$, then by \thmref{thm:mestimator:framework}, there exists a perfect $G$ sampler  that uses $\O{\log n}$ bits of space when $G$ is the Fair estimator.
	
For the Huber measure function, we have $G(x)=\frac{x^2}{2\tau}$ for $|x|\le \tau$ and $G(x)=|x|-\frac{\tau}{2}$ otherwise, where $\tau>0$ is some constant parameter. 
Hence, $G(x)-G(x-1)<1$ and $G(x)>\frac{\tau}{2}\cdot m$, so there exists a perfect $G$ sampler that uses $\O{\log n}$ bits of space when $G$ is the Huber estimator by \thmref{thm:mestimator:framework}.	
\end{proof}

\subsubsection{Matrix Norms on Insertion-Only Streams}\label{sec:matrixnorm}
\seclab{sec:matrix}
We now consider the case of sampling row $\m_i$ from a matrix $\M\in\mathbb{R}^{n\times d}$ with rows $\m_1,\ldots,\m_n\in\mathbb{R}^d$ with probability $\frac{G(\m_i)}{F_G}$ for some function $G$, where we define $F_G=\sum_{j\in[n]}G(\m_j)$. 
We consider a insertion-only stream in the sense that each update in the stream is a non-negative update to some coordinate of $\M$. 

We generalize the approach of \algref{alg:perfect:framework} by first using reservoir sampling to sample an update to a coordinate $c$ to a row $r$ of $\M$. 
We then main a vector $\v$ that consists of all the updates to row $r$ and choose to output $r$ with probability $G(\v+e_c)-G(\v)$, where $e_c$ represents the elementary vector in $\mathbb{R}^d$ with a single $1$ in coordinate $c$ and zeros elsewhere. 

\begin{algorithm}[!htb]
	\caption{Truly perfect $G$-sampler algorithm for vectors and matrices in insertion only streams.}
	\alglab{alg:matrix:framework}
	\begin{algorithmic}[1]
		\Require{A stream of updates $u_1,u_2,\ldots,u_m$, where each $u_i\in[n]\times[d]$ represents a single update to a coordinate of a underlying matrix $\M$, a measure function $G$.}
		\State{Initialize an instance of $\sampler$.}
		\Comment{\algref{alg:sampler}}
		\For{each update $u_t\in[n]\times[d]$}
		\State{Update $\sampler$.}
		\EndFor
		\State{Let $r$ be the row and $c$ be the column sampled by $\sampler$ and let $\v$ be the vector induced by the subsequent updates to row $r$.}
		\State{Let $\zeta$ be a parameter such that $G(\x)-G(\x-e_i)\le\zeta$ for all $x\ge\left(\mathbb{R}^{\ge0}\right)^{d}$, $i\in[d]$.}
		\State{\Return $r$ with probability $\frac{G(\v+e_c)-G(\v)}{\zeta}$.}
	\end{algorithmic}
\end{algorithm}

The correctness of \algref{alg:matrix:framework} follows from a similar proof to that of \thmref{thm:mestimator:framework}.

\begin{theorem}\label{thm:matrix}
\thmlab{thm:matrix:framework}
 Fix any non-negative function  $G: \R^d \to \R_{\geq 0}$ satisfying $G(\vec{0}) = 0$.
Let $\zeta$ be a parameter such that $G(\x)-G(\x-e_i)\le\zeta$ for all $\x\in\left(\mathbb{R}^{\ge0}\right)^{d}$, $i\in[d]$. 
Given a lower bound $\widehat{F_G}$ on $F_G$, then there exists a truly perfect $G$ sampler for an insertion-only stream that succeeds with probability at least $1-\delta$ and uses $\O{\frac{\zeta dm}{\widehat{F_G}}\log n\log\frac{1}{\delta}}$ bits of space.
\end{theorem}
\begin{proof}
The probability that the update to $(r,c)$ is the $j\th$ particular update to row $r$ inside the stream is $\frac{1}{m}$. 
	Let $\x_j$ be the sum of the updates to row $r$ after the $j\th$ update and let $e_{c_j}$ be the coordinate of row $r$ incremented in the $j\th$ update, so that the probability that $i$ is output is 
	\[\sum_j\frac{1}{m}\frac{G(\x_j+e_{c_j})-G(\x_j)}{\zeta}=\sum_j\frac{1}{m}\frac{G(\x_{j-1})-G(\x_j)}{\zeta}\]\[=\frac{G(\m_i)}{\zeta m},\]
	where the final equality results from the observations that $\x_0=\m_r$ and that $G(\mathbf{0})=0$, since $\v$ must be the all zeros vector after the last update to row $r$.  
	Thus conditioned on some row being output, the algorithm outputs each row $i\in[n]$ with probability $\frac{G(f_i)}{F_G}$. 
	We again note that $G(\x_j+e_{c_j})-G(\x_j)\le\zeta$ for all $\x\in\left(\mathbb{R}^{\ge0}\right)^{d}$, $i\in[d]$, so returning $r$ with probability $\frac{G(\x_j+e_{c_j})-G(\x_j)}{\zeta}$ is well-defined.
	
	Therefore, the probability that some row is returned by \algref{alg:matrix:framework} is $\sum_{i\in[n]}\frac{G(\m_i)}{\zeta m}=\frac{F_G}{\zeta m}$, where $F_G=\sum G(\m_i)$. 
	By repeating the sampler $\O{\frac{\zeta m}{F_G}\log\frac{1}{\delta}}$ times, the algorithm successfully outputs a sample $s$ with probability at least $1-\delta$. 
	We again note that although the algorithm does not know the value of $F_G$, it suffices to repeat the sampler $\O{\frac{\zeta m}{\widehat{F_G}}\log\frac{1}{\delta}}$ times for some lower bound $\widehat{F_G}$ on $F_G$. 
	Each instance only requires $\O{d\log n}$ bits to maintain the vector $\v$, assuming $\log m=\O{\log n}$ and each update is bounded by $\poly(n)$, which can be expressed using $\O{\log n}$ bits.  
	Thus the total space used is $\O{\frac{\zeta dm}{\widehat{F_G}}\log n\log\frac{1}{\delta}}$ bits of space.
	\end{proof}

For example, when $G(\x)=\sum_{i\in[d]}|x_i|$, then $F_G$ is the $L_{1,1}$ norm. 
Then $F_G=m$, so that by \thmref{thm:matrix:framework}, so we can sample a row $\m_i$ with probability proportional to its $L_1$ norm, using $\O{d\log n\log\frac{1}{\delta}}$ bits of space.
We can also apply \thmref{thm:matrix:framework} when $G(\x)=\sqrt{\sum_{i\in[d]}x_i^2}$ is the $L_2$ norm of each row, so that $F_G$ is the $L_{1,2}$ norm crucially used in many adaptive sampling techniques (see \cite{MahabadiRWZ20} and references therein). 

%For example, \cite{AndoniDIW09} showed that when $G(\x)$ is a general norm on $\mathbb{R}^d$, then $G(\x)$ can be used to sketch the earth-mover distance.

Finally, we show in \appref{app:strict:turnstile} that our framework can be extended to strict turnstile streams, i.e., \thmref{thm:perfect:multi:pass}.

\section{Applications in Sliding Windows}
\label{sec:apps:sliding}
In this section, we give additional applications of our framework to truly perfect samplers on sliding windows. 
Recall that in the sliding window model, updates $u_1,\ldots,u_m$ to an underlying vector $f\in\mathbb{R}^n$ arrive sequentially as a data stream and the underlying vector $f$ is determined by the most recent $W$ updates $u_{m-W+1},\ldots, u_m$, where $W>0$ is the predetermined window size parameter. 
We assume that $m$ and $W$ are polynomially bounded in $n$, i.e., $\O{\log m}=\O{\log W}=\O{\log n}$. 
%Each update $u_k$ is an ordered pair $(i_k,\Delta_k)$, which changes $f_{i_k}$ by $\Delta_k$ and at the end of the stream $f_i=\{\sum \Delta_k: i_k=i\}$.  
%In \emph{insertion-only} streams, each update $u_k$ satisfies $\Delta_k=1$ whereas no such restriction applies for \emph{turnstile} streams. 
For each update $u_k$, if $k<m-W+1$, we say $u_k$ is \emph{expired}. 
Otherwise if $k\ge m-W+1$, we say $u_k$ is \emph{active}.

\subsection{$M$-estimators on Sliding Windows}
\seclab{sec:sw:mestimators}
In this section, we consider a general paradigm for sampling from general statistical $M$-estimator distributions. 
%We generalize the paradigm from \secref{sec:perfect:lp:large} to sampling from general statistical $M$-estimator distributions. 
Recall that for a given a measure function $G:\mathbb{R}\to\mathbb{R}^{\ge0}$ such that $G(x)=G(-x)$, $G(0)=0$ and $G$ is non-decreasing in $|x|$, we define $F_G=\sum_{i=1}^n G(f_i)$ so that $F_G$ is also implicitly defined by only the most recent $W$ updates. 
Then a truly perfect $M$-estimator sampler outputs index $i\in[n]$ with probability exactly $\frac{G(f_i)}{F_G}$. 

The key argument in \thmref{thm:perfect:lp:large} was that if $c$ instances of the sample $s$ appeared after the initial sample, then $s$ is output with probability proportional to $c^p-(c-1)^p$. 
By a telescoping argument, each index $i$ is sampled with probability proportional to $\sum_{c=1}^{f_i}c^p-(c-1)^p=f_i^p$. 
Since $G$ is non-decreasing in $|x|$, we can generalize to sampling each item with probability proportional to $G(c)-G(c-1)$, rather than $c^p-(c-1)^p$. 
This approach can be simulated in the sliding window model by checking whether $s$ is an active element. 

\begin{algorithm}[!htb]
\caption{Truly perfect $M$-estimator sampler algorithm for the sliding window model on insertion only streams.}
\alglab{alg:sw:perfect:mestimator}
\begin{algorithmic}[1]
\Require{A stream of updates $u_1,u_2,\ldots,u_m$, where each $u_i\in[n]$ represents a single update to a coordinate of the underlying vector $f$, a measure function $G$, and a size $W$ for the sliding window.}
\State{Initialize instances of $\sampler$ every $W$ updates and keep the two most recent instances.}
\Comment{\algref{alg:sampler}}
\For{each update $u_t\in[n]$ with $t\in[m]$}
\State{Update each $\sampler$.}
\EndFor
\State{Let $s$ be the sampled output of a $\sampler$ and let $c$ be the number of times $s$ has appeared afterwards.}
\State{Let $\zeta$ be a parameter such that $G(x)-G(x-1)\le\zeta$ for all $x\ge 1$.}
\If{$s$ was sampled within the last $W$ updates}
\State{\Return $s$ with probability $\frac{G(c+1)-G(c)}{\zeta}$.}
\EndIf
\end{algorithmic}
\end{algorithm}

\begin{restatable}{theorem}{thmswmemstimatorframework}
\thmlab{thm:sw:mestimator:framework}
Let $G$ be a function such that $G(x)-G(x-1)\le\zeta$ for all $x\ge 1$. 
Then there exists a truly perfect $G$ sampler for the sliding window model that succeeds with probability at least $1-\delta$ and uses $\O{\frac{\zeta W}{F_G}\log n\log\frac{1}{\delta}}$ bits of space.
\end{restatable}
\begin{proof}
As in \thmref{thm:perfect:lp:large}, $\sampler$ contains not only active elements, but an additional number of expired elements. 
If $s$ is an active element, then the probability that $s$ is the $j\th$ particular instance of item $i$ inside the window $\frac{1}{W}$. 
Since the number of instances of $i$ appearing after $j$ is $f_i-j+1$ then the probability that $i$ is output is 
\[\sum_{j=1}^{f_i}\frac{1}{W}\frac{G(f_i-j+1)-G(f_i-j)}{\zeta}=\frac{G(f_i)}{\zeta W}.\]
Again we note that $G(f_i-j+1)-G(f_i-j)\le\zeta$ for all $j\in[f_i]$, so returning $s$ with probability $\frac{G(c)-G(c-1)}{\zeta}$ is valid.

Hence if $s$ is an active element, the probability that $s$ is output is at least $\frac{F_G}{\zeta W}$, where $F_G=\sum G(f_i)$. 
Thus by repeating the sampler $\O{\frac{\zeta W}{F_G}\log\frac{1}{\delta}}$ times, the algorithm will output a sample $s$ with probability at least $1-\delta$. 
Moreover, the sample $s$ will equal each index $i\in[n]$ with probability $\frac{G(f_i)}{F_G}$. 

Note that a single instance requires $\O{\log n}$ bits to maintain the counter $c$, assuming $\log W=\O{\log n}$. 
Thus the total space used is $\O{\frac{\zeta W}{F_G}\log n\log\frac{1}{\delta}}$ bits of space.
\end{proof}

Using \thmref{thm:sw:mestimator:framework} and the properties of the $M$-estimators defined in \corref{cor:stream:mestimator}, we have:
\begin{corollary}
\corlab{cor:sw:mestimator}
There exist truly perfect $G$ samplers for the insertion-only sliding window model that succeed with probability at least $1-\delta$ and use $\O{\log n\log\frac{1}{\delta}}$ bits of space when $G$ is the $L_1-L_2$ estimator, the Fair estimator, or the Huber estimator. 
\end{corollary}
We also give truly perfect $L_p$ samplers on sliding windows. 
%We first introduce general techniques for sliding window algorithms. 
Specifically we prove \thmref{thm:perfect:lp:large}, the details of which can be found in \appref{app:perfect:lp:large}.
\section{$F_0$ Sampling}\label{sec:l0}
In this section, we give truly perfect $F_0$ samplers in the insertion-only streaming model. 
In the $F_0$ sampling problem, the goal is to sample a coordinate $i\in[n]$ from a frequency vector of length $n$ such that $\PPr{i=j}=0$ if $f_j=0$ and $\PPr{i=j}=\frac{1}{F_0}$ otherwise, where $F_0=|\{i\in[n],\:\,f_i\neq 0\}|$. 
We cannot immediately apply the framework of \algref{alg:perfect:framework} for $F_0$ sampling without trivializing the space complexity, due to the fact that $F_0$ can be substantially smaller than $m$. 

We first remark that in the random oracle model, where an algorithm is given oracle access to a random hash function $h:[n]\to[0,1]$, the well-known algorithm that outputs the nonzero coordinate $i\in[n]$ of $f_i$ that minimizes $h(i)$ is truly perfect $F_0$ sampler, since each of the $|F_0|$ has probability $\frac{1}{|F_0|}$ of obtaining the minimal hash value for a random hash function.  
The random oracle model allows the generation of, and repeated access to any number of random bits, and in particular it allows for repeated access to the same $\Omega(n)$ random bits without charging the algorithm for storing the random bits. 
\begin{remark}
In the random oracle model, there exists a truly perfect $F_0$ sampler on insertion-only streams that uses $\O{\log n}$ bits of space and constant update time. 
\end{remark}
Note that storing $\Omega(n)$ random bits without the random oracle assumption is not interesting in the streaming model, as it corresponds to storing the entire input, up to logarithmic factors.
We now give a truly perfect $F_0$ sampler on insertion-only streams without the assumption of the random oracle model. 
We store the first $\sqrt{n}$ distinct items and in parallel sample a set $S$ of $\O{\sqrt{n}}$ random items from the universe. 
We maintain the subset $U$ of $S$ that arrive in the stream. 
Now since we store the first $\sqrt{n}$ distinct items, we know whether $F_0\le\sqrt{n}$ or $F_0>\sqrt{n}$. 
If $F_0\le\sqrt{n}$ then our algorithm has collected all items in the stream and can simply output an item uniformly at random. 
Otherwise if $F_0$ is larger than $\sqrt{n}$ at the end of the stream, then since $S$ has size $\O{\sqrt{n}}$ and was generated at random, we have constant probability that some element of $S$ has arrived in the stream. 
Our algorithm can then output a random element of $U$ that has appeared in the stream. 
We give our algorithm in full in \algref{alg:F0:simple}.

\begin{algorithm}[!htb]
\caption{Truly perfect $F_0$ sampler for insertion only streams.}
\alglab{alg:F0:simple}
\begin{algorithmic}[1]
\Require{A stream of updates $u_1,u_2,\ldots,u_m$, where each $u_i\in[n]$ represents a single update to a coordinate of a underlying vector $f$.}
\State{Let $S$ be a random subset of $[n]$ of size $2\sqrt{n}$.}
\State{Let $T$ be the first unique $\sqrt{n}$ coordinate updates to $f$ in the stream.}
\State{Let $U$ be the subset of $S$ that appears in the stream.}
\If{$|T|<\sqrt{n}$}
\State{\Return a random element of $T$}
\ElsIf{$|U|>0$}
\State{\Return a random element of $U$}
\Else
\State{\Return $\fail$}
\EndIf
\end{algorithmic}
\end{algorithm}
 
\begin{theorem}\label{thm:truly:perfectF0}
\thmlab{thm:truly:perfect:F0}
Given $\delta\in(0,1)$, there exists a truly perfect $F_0$ sampler on insertion-only streams that uses $\O{\sqrt{n}\log n\log\frac{1}{\delta}}$ bits of space and constant update time and succeeds with probability at least $1-\delta$. 
Moreover, the algorithm reports $f_i$ along with the sampled index $i\in[n]$. 
\end{theorem}
\begin{proof}
	Consider \algref{alg:F0:simple}. 
	If $F_0<\sqrt{n}$, then all nonzero coordinates of $f$ will be stored in $T$ in the insertion-only streaming model and so a random element of $T$ is a truly perfect $F_0$ sample. 
	Otherwise if $F_0\ge\sqrt{n}$ and the algorithm does not return $\fail$, then it must have output a random item in $U$, which is a subset of $S$.  
	Since $S$ is a subset of $n$ chosen uniformly at random, then a random element of $S$ is a truly perfect $F_0$ sample. 
	
	It remains to analyze the probability that \algref{alg:F0:simple} returns $\fail$ when $F_0\ge\sqrt{n}$, since it will never fail when $F_0<\sqrt{n}$. 
	Our algorithm will only fail if $|U|=0$, so that none of the $2\sqrt{n}$ random items in $S$ appeared in the stream. 
	This can only occur with probability at most $\left(1-\frac{2\sqrt{n}}{n}\right)^{\sqrt{n}}$, which is certainly at most $\frac{1}{e}$ for sufficiently large $n$. 
	Hence by repeating $\O{\log\frac{1}{\delta}}$ times, the algorithm has probability at least $1-\delta$ of success. 
	Then the space required is $\O{\sqrt{n}\log n\log\frac{1}{\delta}}$. 
	
	Finally, note that the algorithm can track the frequency of each element in $U$ and $T$, so the algorithm can also report the frequency $f_i$ corresponding to the sampled index $i$. 
\end{proof}

%Since the algorithm reports $f_i$ for each sampled index $i\in[n]$, we can furthermore obtain a truly perfect $L_p$ sampler. 

By modifying $T$ to be the last unique $\sqrt{n}$ coordinate updates to $f$ in the stream and storing timestamps for each element in $U$, \algref{alg:F0:simple} extends naturally to the sliding window model. 
\begin{corollary}
\corlab{cor:sw:truly:perfect:F0}
Given $\delta\in(0,1)$, there exists a truly perfect $F_0$ sampler in the sliding window model that uses $\O{\sqrt{n}\log n\log\frac{1}{\delta}}$ bits of space and constant update time and succeeds with probability at least $1-\delta$.  
\end{corollary}

Recall that for the Tukey measure function, we have $G(x)=\frac{\tau^2}{6}\left(1-\left(1-\frac{x^2}{\tau^2}\right)^3\right)$ for $|x|\le\tau$ and $G(x)=\frac{\tau^2}{6}$ otherwise, where $\tau>0$ is some constant. 
We can now use our $F_0$ sampler of choice, say $\zsampler$, to obtain a truly perfect sampler for the Tukey measure function by a similar approach to \algref{alg:perfect:framework}. 
Each instance will use a subroutine of $\zsampler$ rather than $\sampler$. 
Now if $c$ is the number of instances of the index output by $\zsampler$ within the window, then we accept $c$ with probability $\frac{G(c)}{G(\tau)}$. 

\begin{theorem}\label{thm:streamtukey}
\thmlab{thm:stream:tukey}
Given $\delta\in(0,1)$, there exists a truly perfect $G$ sampler for the insertion-only streaming model that succeeds with probability at least $1-\delta$ when $G$ is the Tukey estimator. 
The algorithm uses $\O{\log n\log\frac{1}{\delta}}$ bits of space in the random oracle model and $\O{\sqrt{n}\log n\log\frac{1}{\delta}}$ otherwise.
\end{theorem}
\begin{proof}
	Note that only the indices that appear in the stream can be output. 
	The probability that any index $i\in[n]$ that ap	pears in the stream is output is $\frac{1}{F_0}\cdot\frac{G(f_i)}{G(\tau)}$. 
	Then a single instance of the algorithm returns an output with probability $\frac{F_G}{F_0\cdot G(\tau)}\ge\frac{G(1)}{G(\tau)}$. 
	Hence repeating $\O{\log\frac{1}{\delta}}$ times outputs an instance with probability at least $1-\delta$. 
	Since $\zsampler$ requires $\O{\log n}$ space in the random oracle model, then the total space used is $\O{\log n\log\frac{1}{\delta}}$ bits of space. 
	Otherwise by \thmref{thm:truly:perfect:F0}, $\zsampler$ requires $\O{\sqrt{n}\log n}$ space so that the total space used is $\O{\sqrt{n}\log n\log\frac{1}{\delta}}$. 
\end{proof}

We also obtain a truly perfect $G$ sampler for the insertion-only sliding window model when $G$ is the Tukey estimator by a similar approach to \algref{alg:sw:perfect:mestimator}. 
We again use a truly perfect $F_0$ sampling algorithm $\zsampler$ of choice and, if $c$ is the number of instances of the index $s$ output by $\zsampler$ within the window, then we accept $s$ with probability $\frac{G(c)}{G(\tau)}$. 

\begin{theorem}
\thmlab{thm:sw:tukey}
Given $\delta\in(0,1)$, there exists a truly perfect $G$ sampler for the sliding window model that succeeds with probability at least $1-\delta$ when $G$ is the Tukey estimator. 
The algorithm uses $\O{\log n\log\frac{1}{\delta}}$ bits of space in the random oracle model and $\O{\sqrt{n}\log n\log\frac{1}{\delta}}$ otherwise.
\end{theorem}
Finally, we show in \appref{app:strict:turnstile} that the same guarantees of \thmref{thm:truly:perfect:F0} can be extended to strict turnstile streams. 
\section{Conclusion}
Our work shows that surprisingly, truly perfect samplers exist in the insertion-only model with a sublinear amount of memory for a large class of loss functions, a large class of objects (vectors, matrices), and several different models (data stream and sliding window). Establishing better upper bounds, or proving lower bounds is a very intriguing open question. In particular, the usual communication complexity lower bound arguments do not seem to apply, and already for our turnstile lower bound we needed to use fine-grained properties of the equality function. 

\section*{Acknowledgements}
This work was supported by National Science Foundation (NSF) Grant No. CCF-1815840, National Institute of Health (NIH) grant 5R01 HG 10798-2, and a Simons Investigator Award.
\bibliographystyle{alpha}
\bibliography{references}
\appendix
\section{Proofs missing from Section~\ref{sec:apps:sliding}}
\applab{app:perfect:lp:large}
We first describe our truly perfect $L_p$ samplers for sliding windows. 

\paragraph{Smooth Histograms.}
Although the smooth histogram framework cannot be used for $L_p$ sampling, we require the approximation of a number of parameters for which the smooth histogram framework can be used. 
We thus provide the necessary background. 
Braverman and Ostrovsky introduced the smooth histogram framework to give algorithms for solve a large number of problems in the sliding window model, such as $L_p$ norm estimation, longest increasing subsequence, geometric mean estimation, or other weakly additive functions~\cite{BravermanO07}. 
The smooth histogram framework is defined using the following notion of smooth functions, where the notation $B\subseteq_s A$ denotes that a stream $B$ arrives at the end of stream $A$, so that $B$ is a substream of $A$.  
\begin{definition}[Smooth function]
	\cite{BravermanO07}
	A function $f$ is $(\alpha,\beta)$-smooth if for any stream $A$ if 
	\begin{enumerate}
		\item
		$0\le f(A)\le\poly(W)$
		\item
		$f(B)\le f(A)$ for $B\subseteq_s A$
		\item
		For any $0<\eps<1$, there exists $0<\beta\le\alpha<1$ such that if $B\subseteq_s A$ and $(1-\beta)f(A)\le f(B)$, then $(1-\alpha)f(A\cup C)\le f(B\cup C)$ any adjacent substream $C$. 
	\end{enumerate}
\end{definition}
Intuitively, a smooth function requires that once a suffix of a data stream becomes a $(1\pm\beta)$-approximation for a smooth function, then it remains a $(1\pm\alpha)$-approximation of the data stream, regardless of the subsequent updates that arrive in the stream. 

The smooth histogram data structure maintains a number of timestamps throughout the data stream, along with a streaming algorithm that stores a sketch of all elements seen in the stream beginning at each timestamp. 
The smooth histogram maintains the invariant that at most three streaming algorithms output values that are within $(1-\beta)$ of each other, since any two of the sketches would always output values that are within $(1-\alpha)$ afterwards. 
Thus for polynomially bounded monotonic functions, only a logarithmic number of timestamps need to be stored. 
See \figref{fig:hist} for intuition on how the sliding window is sandwiched between two timestamps of the smooth histogram. 
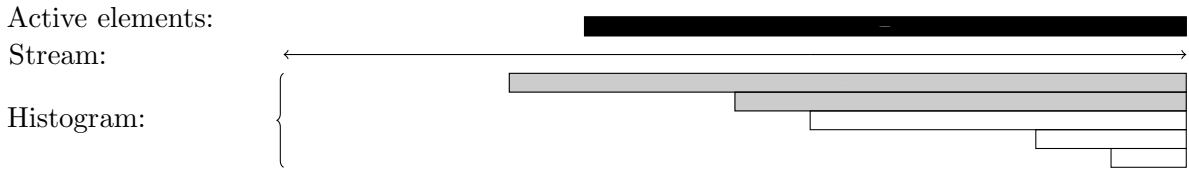
\begin{figure*}[!htb]
	\centering
	\begin{tikzpicture}[scale=1]
	\draw [->] (5,0) -- (10,0);
	\draw [->] (5,0) -- (-2,0);
	\node at (-5,0){Stream:};
	
	\filldraw[shading=radial, inner color = white, outer color = green!50!, opacity=1] (2,0.25) rectangle+(8,0.25);
	\draw (2,0.25) rectangle+(8,0.25);
	\node at (-4.3,0.5){Active elements:};
	
	\filldraw[shading=radial,inner color=white, outer color=gray!90, opacity=0.2] (1,-0.25) rectangle+(9,-0.25);
	\draw (1,-0.25) rectangle +(9,-0.25);
	\filldraw[shading=radial,inner color=white, outer color=gray!90, opacity=0.2] (4,-0.5) rectangle+(6,-0.25);
	\draw (4,-0.5) rectangle +(6,-0.25);
	\draw (5,-0.75) rectangle +(5,-0.25);
	\draw (8,-1) rectangle +(2,-0.25);
	\draw (9,-1.25) rectangle +(1,-0.25);
	
	\draw [decorate,decoration={brace}] (-2,-1.5) -- (-2,-0.25);
	\node at (-4.75,-0.875){Histogram:};
	\end{tikzpicture}
	\caption{Histogram paradigm. Note the first two algorithms sandwich the active elements.}
	\figlab{fig:hist}
\end{figure*}

The smooth histogram has the following properties:
\begin{definition}[Smooth Histogram]
	\cite{BravermanO07}
	Let $g$ be a function that maintains a $(1+\eps)$-approximation of an $(\alpha,\beta)$-smooth function $f$ that takes as input a starting index and ending index in the data stream. 
	The approximate smooth histogram is a structure that consists of an increasing set of indices $X_N=\{x_1,\ldots,x_s=N\}$ and $s$ instances of an algorithm $\Lambda$, namely $\Lambda_1,\ldots,\Lambda_s$ with the following properties:
	\begin{enumerate}
		\item
		$x_1$ corresponds to either the beginning of the data stream, or an expired point.
		\item
		$x_2$ corresponds to an active point. 
		\item
		For all $i<s$, one of the following holds:
		\begin{enumerate}
			\item
			$x_{i+1}=x_i+1$ and $g(x_{i+1},N)<\left(1-\frac{\beta}{2}\right)g(x_i,N)$.
			\item
			$(1-\alpha)g_(x_i,N)\le g(x_{i+1},N)$ and if $i+2\le s$ then $g(x_{i+2},N)<\left(1-\frac{\beta}{2}\right)g(x_i,N)$.
		\end{enumerate}
		\item
		$\Lambda_i=\Lambda(x_i,N)$ maintains $g(x_i,N)$.
	\end{enumerate}
\end{definition}
We define an augmented histogram to be a smooth histogram with additional auxiliary algorithms corresponding to each timestamp. 
\begin{definition}[Augmented Histogram]
	An \emph{augmented histogram} on a data stream $u_1,\ldots,u_t$ is a data structure $\H$ that consists of algorithms $\{\A^{(i)}_1,\ldots,\A^{(i)}_s\}_i$, each with a corresponding time $t_1,\ldots,t_s$. 
	The input to each algorithm $\A^{(i)}_j$ is the sequence of updates $u_{t_j},u_{t_j+1},\ldots,u_t$, so that the input to $\A^{(i)}_j$ is the same as the input to $\A^{(k)}_j$.
\end{definition}

\noindent
The following theorem shows the smoothness of $F_p=\sum_{i=1}^n f_i^p$. 
\begin{theorem}
	\cite{BravermanO07}
	\thmlab{thm:fp:smooth}
	For $p\ge 1$, $F_p$ is $\left(\eps,\frac{\eps^p}{p^p}\right)$-smooth. 
	For $p<1$, $F_p$ is $(\eps,\eps)$-smooth.
\end{theorem}

We run several instances of the same algorithm. 
In each instance, we first use reservoir sampling to sample each item with probability $\frac{1}{2W}$. 
If the stream ends and the sampled item is outside the sliding window, then nothing is output by the algorithm and we move onto the next instance of the algorithm. 
When an item $s$ in the sliding window is selected by the reservoir sampling procedure, we keep a counter $c$ for how many times the item appears afterward. 
We also use a $F_2$ estimation algorithm to find a $2$-approximation $F$ of $\sqrt{F_2}$. 
We then output $s$ with probability $\frac{c^2-(c-1)^2}{F}$. 
Otherwise, if we choose not to output $s$, we again move onto the next instance of the algorithm. 

It can be shown through a straightforward telescoping argument that the probability of outputting each $i\in[n]$ is $\frac{f_i^2}{2FW}$. 
Since $W=F_1$ is greater than $F_2$ by at most a factor of $\sqrt{W}$, the probability that a single instance of the algorithm outputs some index is at least $\frac{1}{2\sqrt{W}}$. 
Hence running $\O{\sqrt{W}}$ instances of the algorithm is a perfect $L_2$ sampler.

We recall the following subroutine for norm estimation in the sliding window model.
\begin{theorem}\cite{BravermanO07}
	\thmlab{thm:lp}
		There exists a sliding window algorithm $\estimate$ that outputs an estimate $f$ such that $f\le L_p\le 2f$, with probability $1-\frac{1}{\poly n}$. 
	The algorithm uses $\O{\log^2 n}$ space.
	%There exists a sliding window algorithm $\estimate$ that takes a parameter $\eps$ and outputs an estimate $f$ such that $f\le L_p\le 2f$, with probability $1-\frac{1}{\poly n}$. 
	%The algorithm uses $\O{\frac{1}{\eps^2}\log^2 n}$ space.
\end{theorem}

\begin{algorithm}[!htb]
	\caption{Perfect $L_p$ sampler for the sliding window model on insertion only streams and $p>1$.}
	\alglab{alg:perfect:lp:insertion:largep}
	\begin{algorithmic}[1]
		\Require{A stream of updates $u_1,u_2,\ldots,u_t$, where each $u_i\in[n]$ represents a single update to a coordinate of the underlying vector $f$, and a size $W$ for the sliding window.}
		\State{Use an augmented histogram $\H$, where $\A^{(1)}_i$ is an instance of $\estimate$ and $\A^{(2)}_i$ is an instance of $\sampler$.}
		\For{each update $u_r$}
		\State{Let $\H$ consist of algorithms $\A^{(j)}_1,\ldots,\A^{(j)}_s$, where $j\in\{1,2\}$.}
		\State{Initialize $\A^{(1)}_{s+1}$ as an instance of $\estimate$ starting with $u_r$.}
		\State{Initialize $\A^{(2)}_{s+1}$ as an instance of $\sampler$ starting with $u_r$.}
		\Comment{\algref{alg:sampler}}
		\State{Set $t_{s+1}=r$.}
		\For{each $1\le i\le s$}
		\State{Update each $\A^{(1)}_i$, $\A^{(2)}_i$ with $u_r$.}
		\EndFor
		\For{each $2\le i\le s-1$}
		\State{Let $N_i$ be the output of $\A^{(1)}_i$.}
		\If{$N_{i-1}\le 2N_{i+1}$}
		\State{Delete $\A^{(1)}_i$, $\A^{(2)}_i$, and $t_i$ from $\H$.}
		\State{Reindex algorithms in $\H$.}
		\EndIf
		\EndFor
		\If{$t_2\le r-W+1$}
		\State{Delete $\A^{(1)}_1$, $\A^{(2)}_1$, and $t_i$ from $\H$.}
		\State{Reindex algorithms in $\H$.}
		\EndIf
		\EndFor
		\State{Let $s$ be the sampled output of $\A^{(2)}_1$ and let $c$ be the number of times $s$ has appeared afterwards.}
		\If{$s$ has timestamp after $t-W+1$}
		\State{Let $F$ be the output of $\A^{(1)}_1$.}
		\State{\Return $s$ with probability $\frac{c^p-(c-1)^p}{pF^{p-1}}$}
		\Comment{$p>1$}
		\EndIf
	\end{algorithmic}
\end{algorithm}

\thmperfectlplarge*
\begin{proof}
	Note that $\sampler$ contains not only the updates in the sliding window, but an additional number of elements, since the substream actually starts before $t-W+1$. 
	Because $F_1=W$ is the number of elements in the substream, then $\sampler$ selects each item inside the window with probability $\frac{1}{F_1}$. 
	Let $\mathcal{E}$ be the event that $\sampler$ selects an element inside the window, which occurs with probability $\frac{W}{F_1}\ge\frac{1}{2}$, since $F_1<2W$. 
	Conditioned on $\mathcal{E}$, then $\sampler$ selects each item of the window with probability $\frac{1}{W}$. 
	Hence if $i\in[n]$ appears $f_i$ times inside the window, then $\sampler$ outputs $i$ with probability $\frac{f_i}{W}$, conditioned on $\mathcal{E}$. 
	
	The probability that the $j\th$ particular instance of $i$ inside the window is selected is $\frac{1}{W}$, conditioned on $\mathcal{E}$. 
	Moreover for $p>1$, the number of instances of $i$ appearing after $j$, inclusive, is $f_i-j+1$ so the probability that $i$ is output is 
	\[\sum_{j=1}^{f_i}\frac{1}{W}\frac{(f_i-j+1)^p-(f_i-j)^p}{pF^{p-1}}=\frac{f_i^p}{pWF^{p-1}},\]
	where $F_p^{\frac{1}{p}}<F<2F_p^{\frac{1}{p}}$ by \thmref{thm:fp:smooth}, with $F_p=\sum f_i^p$. 
	Note that $(f_i-j+1)^p-(f_i-j)^p\le pF^{p-1}$ for all $j\in[f_i]$, so returning $s$ with probability $\frac{c^p-(c-1)^p}{pF^{p-1}}$ is a valid procedure.
	
	Since $F_p^{\frac{1}{p}}\le F_1=W\le W^{1-\frac{1}{p}}F_p^{\frac{1}{p}}$ for $p>1$, then the probability that some sample is output is at least 
	\[\sum\frac{f_i^p}{pWF^{p-1}}\ge\frac{F_p}{p\left(W^{1-\frac{1}{p}}F_p^{\frac{1}{p}}\right)\left(2F_p^{\frac{1}{p}}\right)^{p-1}}\ge\frac{1}{p2^{p-1} W^{1-\frac{1}{p}}}.\] 
	Thus by repeating the sampler $\O{W^{1-\frac{1}{p}}}$ times, the algorithm will output a sample $s$ with probability at least $\frac{2}{3}$. 
	Moreover, the sample $s$ will equal each index $i\in[n]$ with probability $\frac{f_i^p}{F_p}$. 
	Since each sampler requires $\O{\log W+\log n}$ bits of space, then the space complexity follows, under the assumption that $\O{\log n}=\O{\log W}$. 
\end{proof}

Jayaram and Woodruff~\cite{JayaramW18} give a perfect $L_p$ sampler in the streaming model that uses $\O{\log^2 n\log\frac{1}{\delta}}$ space for $0<p<1$ to obtain $1-\delta$ probability of success. 
Their algorithm takes the underlying universe and duplicates each item $\poly(n)$ times. 
As the stream arrives, an update to each element of the original stream is transformed into updates to each of the $\poly(n)$ duplicated items. 
A linear transformation is then performed by scaling each duplicated item by the inverse of an exponential random variable. 
A sketch is maintained throughout the stream to output the item with the largest frequency at the end of the stream if it is sufficiently large. 
In particular, if $z_0$ is the frequency of the scaled duplicated item with the largest frequency and $z'$ is the frequency vector of the scaled duplicated items \emph{excluding} $z_0$, then the item will be reported if $|z_0|>20\norm{z'}_2$. 
Otherwise, the algorithm does not output anything, but multiple instances of the algorithm are run in parallel to ensure that some instance reports an item with probability at least $1-\delta$.

\section{Fast Perfect $L_p$ Sampler for $0<p<1$ on Sliding Windows}
\seclab{sec:perfect:lp:small}
In this section, we give a construction of a perfect $L_p$ sampler (but not truly perfect) for $p\in(0,1)$ on sliding windows. 
As a specific case, it also provides a perfect $L_p$ sampler in insertion-only streaming model that has faster update time than existing constructions, e.g.~\cite{JayaramW18}.

Recall that an exponential random variable $E$ is parametrized by a rate $\lambda>0$ if it has cumulative distribution function $\PPr{E<x}=1-e^{-\lambda x}$. 
\begin{fact}[Scaling of exponentials]
\factlab{fact:exp}
Let $E$ be an exponential random variable with rate $\lambda>0$ and let $\alpha>0$. 
Then $\alpha E$ is an exponential random variable with rate $\frac{\lambda}{\alpha}$.  
\end{fact}

\begin{algorithm}[!htb]
\caption{Perfect $L_p$ sampler for the sliding window model on insertion only streams and $p<1$.}
\alglab{alg:perfect:lp:insertion:smallp}
\begin{algorithmic}[1]
\Require{A stream of updates $u_1,u_2,\ldots,u_t$, where each $u_i\in[n]$ represents a single update to a coordinate of the underlying vector $f$, and a size $W$ for the sliding window.}
\State{Let $\alpha$ be a sufficiently large constant so that $n^\alpha\gg\poly(n,W)$.}
%\Comment{\thmref{thm:lp}}
\For{$k=1$ to $k=\alpha\log n$}
\State{$\mathcal{S}_k\gets\emptyset$}
\EndFor
\State{Let $e_{i,j}$ be an exponential random variable for each $i\in[n],j\in[n^c]$.}
\For{each update $u_r$}
\For{$i\in[n^c]$}
\For{$k=1$ to $k=\alpha\log n$}
\State{Generate $\frac{1}{e_{u_r,i}^{1/p}}$ instances of duplicated variable $z_{u_r,i}$.}
\If{$|\mathcal{S}_k|<400\alpha c\log n$}
\State{Add each instance into $\mathcal{S}_k$ with probability $\frac{100c\log n}{2^k}$ with timestamp $r$.}
\EndIf
\State{Delete elements of $\mathcal{S}_k$ with timestamp less than $r-W$.}
\EndFor
\EndFor
\EndFor
\State{Let $\widehat{f}$ be a $2$-approximation to the number of instances of duplicated variables in the active window.}
\Comment{\thmref{thm:lp}}
\State{Let $k$ be the integer with $2^k\le\widehat{f}<2^{k+1}$}
\If{there exists $(i,j)$ such that $z_{i,j}$ forms a majority of $\mathcal{S}_k$}
\State{\Return $i$.}
\Else
\State{Output FAIL.}
\EndIf
\end{algorithmic}
\end{algorithm}

For a vector $z$, we define the anti-ranks to describe the indices of $z$ whose corresponding frequencies are non-increasing in magnitude. 
\begin{definition}
Let the anti-ranks $D(\cdot)$ be defined so that $|z_{D(1)}|\ge |z_{D(2)}|\ge\ldots\ge |z_{D(n^c)}|$. 
We use the notation $z_{-D(1)}$ to denote the vector $z$ whose largest entry has been set to zero, i.e., $z_{D(1)}=0$. 
\end{definition}

\begin{lemma}
\cite{Nagaraja06}
\lemlab{lem:exp:stability}
Let $z_1,\ldots,z_n$ be independent exponential random variables, so that $z_i$ has rate $\lambda_i>0$ for each $i\in[n]$. 
Then for any $i\in[n]$, we have
\[\PPr{D(1)=i}=\frac{\lambda_i}{\sum_{j\in[n]}\lambda_j}.\]
\end{lemma}

We abuse notation and represent $z$ in \algref{alg:perfect:lp:insertion:smallp} as both a frequency vector with $n^{c+1}$ coordinates as well as a doubly indexed set $z_{i,j}$ with $i\in[n]$ and $j\in[n^c]$, to denote that each coordinate $i\in[n]$ is duplicated $n^c$ times. 
To analyze \algref{alg:perfect:lp:insertion:smallp}, we first recall the following lemma from \cite{JayaramW18}, which decomposes the value of the maximum coordinate of the duplicated vector into a large component that is independent of the index that achieves the max and a negligible component that depends on the index. 
\begin{lemma}\cite{JayaramW18}
\lemlab{lem:u:v}
For each $1\le k<n^c-n^{9c/10}$, $p\in(0,2]$ and $\nu\ge n^{-c/60}$, we have with probability $1-\O{e^{-n^{c/3}}}$ that $|z_{D(k)}|=U_{D(k)}(1+V_{D(k)})$, for $|V_{D_k}|=\O{\nu}$ and 
\[U_{D_k}=\left[\left(1\pm\O{n^{-c/10}}\right)\sum_{\tau=1}^k\frac{E_{\tau}}{\Ex{\sum_{j=1}^{n^c} F_{D(j)}^p}}\right]^{-1/p},\]
where the $E_{\tau}$ are identically independently distributed exponential random variables that are independent of the value of $D(k)$ and $F$ is the underlying frequency vector on the universe of $n^c$ duplicated items. 
\end{lemma}
The following lemma from \cite{JayaramW18} states that with constant probability, the max coordinate of $z$ in \algref{alg:perfect:lp:insertion:smallp} will dominate the $p$-norm of $z$. 
\begin{lemma}\cite{JayaramW18}
\lemlab{lem:z:superheavy}
For $p<2$, $z_{D(1)}>20\norm{z_{-D(1)}}_p$ with constant probability. 
\end{lemma}

We now show that the probability that \algref{alg:perfect:lp:insertion:smallp} fails only negligibly depends on the index that achieves the max. 
This negligible difference can be absorbed into an $\frac{1}{\poly(n)}$ additive component for the sampling probabilities. 
\begin{lemma}
\lemlab{lem:fail}
Let $\mathcal{E}$ denote the event that \algref{alg:perfect:lp:insertion:smallp} fails. 
Then $\PPr{\mathcal{E}\,|\,z_{D(1)}}=\PPr{\mathcal{E}}+\frac{1}{\poly(n)}$.
\end{lemma}
\begin{proof}
Consider \algref{alg:perfect:lp:insertion:smallp}. 
We first analyze the probability that the algorithm fails, conditioned on a given value of $D(1)$. 
By \lemref{lem:u:v}, we can rewrite $|z_{D(1)}|=U_{D(1)}(1+V_{D(1)})$, where $U_{D_1}$ is independent of the value of $D(1)$ and $|V_{D_1}|=\O{\nu}$. 
Thus the probability that the algorithm fails conditioned on the value of $D(1)$ is at most an additive $\O{\nu}$ amount from the probability that the algorithm fails. 
Hence, it suffices to set $\nu=\frac{1}{\poly(n)}$. 

To evaluate the probability that the algorithm fails, note that \lemref{lem:z:superheavy} implies $|z_{D(1)}|>20\norm{z_{-D(1)}}_p>20\norm{z_{-D(1)}}_1$ with constant probability for $p<1$. 
Let $2^k\le\norm{z}_1\le 2^{k+1}$ and let $\beta=\frac{\norm{z}_1}{2^k}\ge1$. 
If each element is sampled with probability $\frac{100c\log n}{2^k}$, then $90c\beta\log n$ instances of $D(1)$ will be sampled in expectation. 
By Chernoff bounds, at least $60c\beta\log n$ samples of $D(1)$ will be drawn with probability at least $1-n^{-5c}$. 
On the other hand, at most $10c\beta\log n$ instances of other elements will be sampled in expectation. 
The probability less than $40c\beta\log n$ instances of other elements are sampled is at least $1-n^{-5c}$. 
Thus $D(1)$ will be output with probability at least $1-n^{-4c}$. 
\end{proof}

We highlight that \lemref{lem:fail} implicitly but crucially uses the fact that each sample is unbiased. 
That is, we cannot blindly interpret \lemref{lem:z:superheavy} as claiming that it suffices to find the heavy-hitter that dominates $p$-norm of the duplicated frequency vector. 
Notably, na\"{i}vely using a sliding window heavy-hitter algorithm such as \cite{BravermanGLWZ18} to identify the index $i\in[n]$ that achieves the max does not work because these algorithms are biased against the items appearing near the boundary of the sliding window.  
For example, any heavy-hitter algorithm including a recently expired insertion to coordinate $i$ is \emph{more likely} to identify $i$ as a heavy-hitter, so that the probability of failure is not independent of which coordinate achieves the max since it may be non-negligibly lower for $i$. 
Similarly, any heavy-hitter algorithm that excludes an active insertion to coordinate $i$ is less likely to identify $i$ as a heavy-hitter and again, the probability of failure is not independent of which coordinate achieves the max since it may be non-negligibly higher for $i$. 
Moreover, if a reduction from \lemref{lem:z:superheavy} to the heavy-hitters problem in the sliding window were immediately true, then it would be possible to use a heavy-hitter sliding window algorithm such as \cite{BravermanGLWZ18} to do perfect $L_p$ sampling in $\polylog(n)$ space.
 %which contradicts \thmref{thm:lb:insertion}. 

Since \lemref{lem:fail} states that the value of $z_{D(1)}$, i.e., the coordinate that achieves the max after the scaling of the exponential random variables, only affects the failure probability of the algorithm by a negligible amount, it remains to show that each coordinate $i$ achieves the max with probability $\frac{|f_i|^p}{\sum_{j\in[n]}|f_j|^p}$. 

\begin{theorem}
\thmlab{thm:perfect:lp:small}
For $p<1$, there exists a perfect $L_p$ sampler for the sliding window model with insertion-only updates that uses $\O{\log^3 W}$ bits of space.
\end{theorem}
\begin{proof}
Let $f\in\mathbb{R}^n$ be the underlying frequency vector implicitly defined by the sliding window. 
Each $i\in[n]$ is associated with $n^c$ variables $z_{i,j}$, where $j\in[n^c]$. 
Since $i$ is updated $f_i$ times in the active window, then we have $z_{i,j}\sim\frac{f_i}{E_{i,j}^{1/p}}$, where $E_{i,j}$ is an exponential random variable. 
Thus by \factref{fact:exp}, each $|z_{i,j}|^{-p}$ is an exponential random variable with rate $|f_i|^p$. 
By \lemref{lem:exp:stability}, the probability that $z_{i,j}$ is the largest scaled variable across all $i\in[n]$, $j\in[n^c]$ is $\frac{|f_i|^p}{n^c\cdot\sum_{k\in[n]}|f_k|^p}$. 
Hence for a fixed $i\in[n]$, the probability that some variable $z_{i,j}$ achieves the max for some $j\in[n^c]$ is 
\[\frac{n^c\cdot|f_i|^p}{n^c\cdot\sum_{k\in[n]}|f_k|^p}=\frac{|f_i|^p}{\sum_{k\in[n]}|f_k|^p}.\]
Now conditioned on some $z_{i,j}$ achieving the max, the probability that \algref{alg:perfect:lp:insertion:smallp} succeeds and correctly outputs $i$ is only perturbed by an additive $\frac{1}{\poly(n)}$ by \lemref{lem:fail}. 
Thus for each $i\in[n]$, the probability that \algref{alg:perfect:lp:insertion:smallp} outputs $i$ is 
\[\frac{|f_i|^p}{\sum_{k\in[n]}|f_k|^p}+\frac{1}{\poly(n)}.\]
Hence, \algref{alg:perfect:lp:insertion:smallp} is a perfect $L_p$ sampler. 

To analyze the space complexity, observe that \algref{alg:perfect:lp:insertion:smallp} maintains $\O{\log n}$ samples in each of the sets $\mathcal{S}_1,\ldots,\mathcal{S}_{\alpha\log n}$. 
Each sample uses $\O{\log n}$ bits to store. 
Hence, these components of the algorithm use $\O{\log^3 n}$ bits of space in total and it remains to derandomize the exponential random variables, which we argue below. 
\end{proof}

\paragraph{Derandomization of the Algorithm.}
\cite{JayaramW18} used a combination of Nisan's pseudorandom generator (PRG)~\cite{Nisan92} and a PRG of~\cite{GopalanKM18} that fools a certain class of Fourier transforms to develop an efficient PRG that fools half-space testers. 
Using this half-space tester, \cite{JayaramW18} showed that any streaming algorithm that only stores a random sketch $\mathbf{A}\cdot f$, with bounded independent and identically distributed entries that can be efficiently sampled, on an input stream whose intermediate underlying vectors have polynomially bounded entries can be efficiently derandomized. 
This suffices for the perfect $L_p$ sampler in the streaming model, which uses a series of linear sketches to sample an output. 
\algref{alg:perfect:lp:insertion:smallp} uses a sampling based approach instead of linear sketches, so we must repurpose the derandomization of the perfect sampler from \cite{JayaramW18}. 
Fortunately, we argue that a simple application of Nisan's PRG suffices to derandomize our algorithm. 

\begin{theorem}[Nisan's PRG]
\cite{Nisan92}
\thmlab{thm:nisan}
Let $\mathcal{A}$ be an algorithm that uses $S=\Omega(\log n)$ space and $R$ random bits. 
Then there exists a pseudorandom generator for $\mathcal{A}$ that succeeds with high probability and runs in $\O{S\log R}$ bits. 
\end{theorem}

Recall that Nisan's PRG can be viewed as generating a stream of pseudorandom bits in a read-once tape that can be used to generate random variables to fool a small space tester. 
However, an input tape that can only be read once cannot be immediately given to algorithm to generate the exponential random variables $e_{i,j}$ and subsequently discarded because the value assigned to each $e_{i,j}$ must be consistent whenever the coordinate $i$ is updated. 
Instead, we use the standard reordering trick to derandomize using Nisan's PRG. 

For any fixed randomness $\R$ for the exponential random variables, let $\T_{\R}$ be the tester that tests whether our $L_p$ sampler would output an index $i\in[n]$ if $\R$ is hard-coded into the tester and the random bits for the sampling procedures arrive in the stream. 
Specifically, we define $\T_{\R}(i,\S,\A_1)=1$ if the algorithm with access to independent exponential random variables outputs $i$ on stream $\S$ and $\T_{\R}(i,\S,\A_1)=0$ otherwise. 
Similarly, we define $\T_{\R}(i,\S,\A_2)=1$ if using Nisan's PRG on our algorithm outputs $i$ on stream $\S$ and $\T_{\R}(i,\S,\A_2)=0$ otherwise. 

Let $\S_1$ be any fixed input stream and $\S_2$ be an input stream in which all updates to a single coordinate of the underlying frequency vector arrive consecutively. 
Observe that using Nisan's PRG on the algorithm suffices to fool $\A_{\R}$ on $\S_2$ from an algorithm with access to independent exponential random variables. 
That is, for all $i\in[n]$, we have
\[\left|\PPr{\T_{\R}(i,\S_2,\A_1)=1}-\PPr{\T_{\R}(i,\S_2,\A_2)=1}\right|=\frac{1}{\poly(n)}.\]
On the other hand, the order of the inputs does not change the distribution of the outputs of the idealized process, so that
\[\PPr{\T_{\R}(i,\S_1,\A_1)=1}=\PPr{\T_{\R}(i,\S_2,\A_1)=1}.\]
Similarly, the order of the inputs does not change the distribution of the outputs of the algorithm following Nisan's PRG, so that
\[\PPr{\T_{\R}(i,\S_1,\A_2)=1}=\PPr{\T_{\R}(i,\S_2,\A_2)=1}.\]
Hence, we have 
\[\left|\PPr{\T_{\R}(i,\S_1,\A_1)=1}-\PPr{\T_{\R}(i,\S_1,\A_2)=1}\right|=\frac{1}{\poly(n)}.\]
In other words, using Nisan's PRG on the algorithm on the $\S_1$ suffices to fool $\A_{\R}$ from an algorithm with access to independent exponential random variables. 

Since our algorithm uses $\O{\log^3 n}$ bits of space and $\O{n}$ bits of randomness, then a na\"{i}ve application of Nisan's PRG would derandomize our algorithm using $\O{\log^4 n}$ bits of space by \thmref{thm:nisan}. 
Instead, we note that only the samples in set $S_k$ with $2^k\le\norm{z}_1\le 2^{k+1}$ are tested by each our algorithm. 
Thus there exists a space $\O{\log^2 n}$ tester for our algorithm, which combined with Nisan's PRG, yields a $\O{\log^3 n}$ space derandomization of \algref{alg:perfect:lp:insertion:smallp}, by \thmref{thm:nisan}. 

\subsection{Insertion-Only Streams}
\begin{algorithm}[!htb]
\caption{Perfect $L_p$ sampler for the streaming model on insertion only streams and $p<1$.}
\alglab{alg:perfect:lp:misragries}
\begin{algorithmic}[1]
\Require{A stream of updates $u_1,u_2,\ldots,u_t$, where each $u_i\in[n]$ represents a single update to a coordinate of the underlying vector $f$, and a size $W$ for the sliding window.}
\State{Let $e_{i,j}$ be an exponential random variable for each $i\in[n],j\in[n^c]$.}
\For{each update $u_r$}
\For{$i\in[n^c]$}
\State{Insert $\frac{1}{e_{u_r,i}^{1/p}}$ instances of duplicated variable $z_{u_r,i}$ into stream $S'$.}
\State{Run $\misragries$ on stream $S'$.} 
\EndFor
\EndFor
\State{Let $m$ be the stream length of $S'$}
\If{$\misragries$ reports an item $i$ with frequency at least $\frac{1}{2}m$.}
\State{\Return $i$.}
\Else
\State{Output FAIL.}
\EndIf
\end{algorithmic}
\end{algorithm}

\begin{theorem}
For $p<1$, there exists a perfect $L_p$ sampler for the streaming model with insertion-only updates that uses $\O{\log n}$ bits of space.
\end{theorem}
\begin{proof}
Consider \algref{alg:perfect:lp:misragries}. 
We fix the value of $D(1)$ and analyze the probability that the algorithm fails. 
\lemref{lem:u:v} again implies that we can rewrite $|z_{D(1)}|=U_{D(1)}(1+V_{D(1)})$, where $U_{D_1}$ is independent of the value of $D(1)$ and $|V_{D_1}|=\O{\nu}$. 
Hence, the value of $D(1)$ only perturbs the probability that the algorithm fails by at most an additive $\O{\nu}$. 
Thus we set $\nu=\frac{1}{\poly(n)}$ and absorb the additive error into the $\frac{1}{\poly(n)}$ sampling error. 

To evaluate the probability that the algorithm fails, note that \lemref{lem:z:superheavy} implies $z_{D(1)}>20\norm{z_{-D(1)}}_p>20\norm{z_{-D(1)}}_1$ with constant probability for $p<1$. 
Now conditioned on $z_{D(1)}>20\norm{z_{-D(1)}}_1$, a $\misragries$ data structure with $\eps=\frac{1}{100}$ in \thmref{thm:misragries} will always include $D(1)$ in the list of indices. 
Moreover, $\misragries$ will output an estimated frequency for $z_{D(1)}$ larger than $\frac{1}{2}\norm{z}_1$. 
Hence, $D(1)$ will always be output by the algorithm. 
Since $\misragries$ with $\eps=\O{1}$ uses $\O{\log n}$ bits of space, then the algorithm uses $\O{\log n}$ bits of space in total. 
\end{proof}

\subsection{Fast Update Time}
\seclab{sec:fast:time}
Observe that we expect to sample $\O{\log n}$ elements at all times within the stream. 
Hence, it is wasteful to generate all variables $z_{u_r,j}$ for $j\in[n^c]$ each time a new update $u_r$ arrives. 
Instead, we use the fact that the sum of exponential random variables converges to a $p$-stable distribution.

\begin{theorem}
\cite{Hall81}
\thmlab{thm:e:pstable}
Let $e_1,\ldots,e_{n^c}$ be exponential random variables with rate $1$. 
Let $\beta=\frac{1}{2p}-\frac{1}{2}$ for $p<1$ and $\beta=\frac{1}{p}-\frac{1}{2}$ for $1<p<2$. 
Then
%\[\PPr{\sum_{i=1}^{n^c}\frac{1}{e_i^{1/p}}\le n^{c/p}x}=A_1(x)+\frac{1}{n^c}A_2(x)+E(x),\]
%where $A_1(x),A_2(x)$ are the probability distribution functions of $p$-stable random variables whose characteristic function can be explicitly computed and $E(x)=o\left(\frac{1}{n^c}\right)$. 
\[\PPr{\sum_{i=1}^{n^c}\frac{1}{e_i^{1/p}}\le n^{c/p}x}=A_1(x)+E(x),\]
where $A_1(x)$ is the probability distribution functions of a $p$-stable random variable whose characteristic function can be explicitly computed and $E(x)=\O{\frac{1}{n^{c\beta}}}$. 
\end{theorem}

Now instead of generating $\frac{1}{e_{i,j}^{1/p}}$ for each $i\in[n]$ and $j\in[n^c]$, we can instead generate a $p$-stable random variable $C_i$ for each $i\in[n]$ according to \thmref{thm:e:pstable} that is approximately $\sum_{j=1}^{n^c}\frac{1}{e_{i,j}^{1/p}}$. 
We can then sample on the variables $C_i$ rather than the individual $e_{i,j}$. 
When a variable $C_i$ is sampled, we compute which variables $e_{i,j}$ the sampled variable $C_i$ corresponds to at the end of the stream, using $\O{n^c}$ post-processing time. 

\lemref{lem:z:superheavy} observes that for $|z_{D(1)}|>20\norm{z_{-D(1)}}_p$ with constant probability for any $p<2$. 
In particular for $p<1$, then $|z_{D(1)}|>20\norm{z_{-D(1)}}_p>20\norm{z_{-D(1)}}_1$, so $z_{D(1)}$ is a heavy-hitter of the frequency vector $F$ of all $z_{i,j}$. 
Thus $|C_{D(1)}|>20\norm{C_{-D(1)}}_1$ so that $C_{D(1)}$ is a heavy-hitter of the frequency vector $C$ and so an instance of CountMin that finds the $\O{1}$ heavy-hitters of $C$ with probability at least $1-\frac{1}{n^c}$ will find $C_{D(1)}$ and identify it as the maximal element. 
The algorithm can then report $D(1)$. 

Similarly, for $1<p<2$, we have $|z_{D(1)}|>20\norm{z_{-D(1)}}_2$, so $z_{D(1)}$ is an $L_2$ heavy-hitter of the frequency vector $F$ of all $z_{i,j}$. 
Now we have $C_{D(1)}>20\norm{C_{-D(1)}}_2$ so that $C_{D(1)}$ is a heavy-hitter of the frequency vector $C$ and so an instance of CountMin that finds the $\O{1}$ heavy-hitters of $C$ with probability at least $1-\frac{1}{n^c}$ will find $C_{D(1)}$ and identify it as the maximal element.

A similar argument can be used to derandomize the fast update time algorithm, by replacing the hard-coded randomness for the independent random variables in the tester with the hard-coded randomness for the independent $p$-stable random variables. 

\begin{corollary}
\corlab{cor:perfect:smallp:fast}
For $p<1$, there exists a perfect $L_p$ sampler for the streaming model with insertion-only updates that uses $\O{\log n}$ bits of space, $\polylog(n)$ update time, and $\poly(n)$ post-processing time. 
\end{corollary}
\section{Truly Perfect Sampling on Random Order Streams}
\applab{app:randomorder}
In this section, we first consider truly perfect $L_2$ sampling on random order streams using $\O{\log^2 n}$ bits of space and generalize this approach to truly perfect $L_p$ samplers on random order streams for integers $p>2$ using $\O{W^{1-\frac{1}{p-1}}\log n}$ bits of space. 
For $p=2$, the idea is to consider two adjacent elements and see if they collide. 
Intuitively, an arbitrary position in the window is item $i$ with probability $\frac{f_i}{W}$ due to the random order of the stream. 
Then the next position in the window is also item $i$ with probability $\frac{f_i-1}{W-1}$. 
The probability the two positions are both $i$ is $\frac{f_i(f_i-1)}{W(W-1)}$, which is not quite the right probability. 

Thus, we ``correct'' this probability by sampling item $i$ in a position with probability $\frac{1}{W}$. 
Otherwise, with probability $1-\frac{1}{W}$, we sample item $i$ if the item is in the next position as well. 
Now the probability of sampling $i$ on the two adjacent elements is $\frac{1}{W}\frac{f_i}{W}+\frac{W-1}{W}\frac{f_i}{W}\frac{f_i-1}{W-1}=\frac{f_i^2}{W^2}$. 
Hence if these two positions outputs some item, then the item is drawn from the correct distribution. 
Moreover, since we maintain the items and their positions, we can expire our samples whenever the items become expired. 
We can also output a random sample in the case there are multiple samples. 
We perform this procedure over all disjoint adjacent pairs in the stream and show that the algorithm will report some sample with constant probability, since we show that the expected number of samples and the squared expected number of samples is within a constant factor.
We give the algorithm in \algref{alg:perfect:l2:random}.

\begin{remark}
We remark that although our results are presented in the sliding window model, they can naturally apply to random-order insertion-only streams as well.
\end{remark}

\begin{algorithm}[!htb]
\caption{Truly perfect $L_2$ sampler algorithm for the sliding window model on random order streams.}
\alglab{alg:perfect:l2:random}
\begin{algorithmic}[1]
\Require{A stream of updates $u_1,u_2,\ldots,u_t$, a size $W$ for the sliding window. 
Let each $u_i\in[n]$ represent a single update to a coordinate of the underlying vector $f$.}
\State{$\mathcal{S}\gets\emptyset$}
\For{each update $u_{2i-1}$ and $u_{2i}$}
\State{Draw $\phi\in[0,1]$ uniformly at random.}
\If{$\phi<\frac{1}{W}$}
\Comment{With probability $\frac{1}{W}$}
\State{$\mathcal{S}\gets\mathcal{S}\cup(u_{2i-1},2i-1)$.}
\Else
\Comment{With probability $1-\frac{1}{W}$}
\If{$u_{2i-1}=u_{2i}$}
\State{$\mathcal{S}\gets\mathcal{S}\cup(u_{2i-1},2i-1)$.}
\EndIf
\EndIf
\If{$(u_j,j)\in\mathcal{S}$ with $j\le 2i-W$}
\Comment{$(u_j,j)$ has expired}
\State{Delete $(u_j,j)$ from $\mathcal{S}$}
\EndIf
\State{Let $C$ be a sufficiently large constant so that $n^C>W$.}
\If{$|\mathcal{S}|>2C\log n$}
\State{Delete $C\log n$ elements uniformly at random from $\mathcal{S}$.}
\EndIf
\EndFor
\State{Return $(u_j,j)\in\mathcal{S}$ uniformly at random.}
\end{algorithmic}
\end{algorithm}

\begin{lemma}
\lemlab{lem:l2:sample:prob}
For each pair $u_{2i-1}$ and $u_{2i}$ with $2i-1>t-W$, the probability that \algref{alg:perfect:l2:random} samples $j\in[n]$ is $\frac{f_j^2}{W^2}$.
\end{lemma}
\begin{proof}
The probability that $u_{2i-1}$ is $\frac{f_j}{W}$ and $u_{2i-1}$ is sampled with probability $\frac{1}{W}$. 
Otherwise, $u_{2i}=u_{2i-1}=j$ with probability $\frac{f_j(f_j-1)}{W(W-1)}$. 
Hence, the probability that \algref{alg:perfect:l2:random} samples $j$ is 
\[\frac{1}{W}\frac{f_j}{W}+\frac{W-1}{W}\frac{f_j(f_j-1)}{W(W-1)}=\frac{f_j^2}{W^2}.\]
\end{proof}

\noindent
We use the following formulation of the Paley-Zygmund Inequality.
\begin{theorem}[Paley-Zygmund Inequality]
Let $X\ge 0$ be a random variable with finite variance and $0\le\delta\le 1$. 
Then
\[\PPr{X>\delta\Ex{X}}\ge(1-\delta)^2\frac{\Ex{X}^2}{\Ex{X^2}}.\]
\end{theorem}

\begin{lemma}
\lemlab{lem:l2:output}
\algref{alg:perfect:l2:random} outputs some sample with probability at least $\frac{2}{3}$.
\end{lemma}
\begin{proof}
By \lemref{lem:l2:sample:prob}, the probability that \algref{alg:perfect:l2:random} samples some $j\in[n]$ for a pair $u_{2i-1}$ and $u_{2i}$ with $2i-1>t-W$ is $\frac{F_2}{W^2}$, where $F_2=\sum_{j=1}^n f_j^2$. 
We suppose $t$ and $W\ge 4$ are even for the purpose of presentation, but a similar argument suffices for other parities of $t$ and $W$. 
Let $X_1,\ldots,X_{\frac{W}{2}}$ be indicator variables so that $X_i=1$ if the $i\th$ disjoint consecutive pair in the window $u_{t-W+2i-1},u_{t-W+2i}$ produced a sample and $X_i=0$ otherwise and let $X=\sum X_i$. 
Thus the expected number of samples $X$ inserted into $\mathcal{S}$ is 
\[\Ex{X}=\sum\Ex{X_i}=\frac{W}{2}\frac{F_2}{W^2}=\frac{F_2}{2W}.\]

Moreover, we have $\Ex{X^2}=\sum\Ex{X_iX_j}$. 
For $i=j$, $\Ex{X_iX_j}=\frac{F_2}{W^2}$. 
For $i\neq j$, consider the probability the $j\th$ disjoint pair is also sampled conditioned on $X_i=1$. 
Formally, if $\tilde{f}_k$ is the frequency of $k$ in the window excluding the $i\th$ pair, then the probability that the $j\th$ disjoint pair samples $k\in[n]$ is $\frac{\tilde{f}_k^2}{(W-2)^2}$ by a similar argument to \lemref{lem:l2:sample:prob}. 
Since $f_k\ge\tilde{f}_k$ and $\sum\tilde{f}_k=-2+\sum f_k$, then $\Ex{X_j|X_i=1}$ is within a constant factor of $\frac{F_2}{W^2}$. 
Hence, $\Ex{X^2}$ is within a constant factor of $\Ex{X}^2$ so then \algref{alg:perfect:l2:random} outputs some sample with constant probability by the Paley-Zygmund inequality.
\end{proof}

We now prove \thmref{thm:small:p}.
%\begin{theorem}
%There exists a one-pass sliding window algorithm for random order insertion only streams that outputs index $i\in[n]$ with probability $\frac{f_i^2}{F_2}$ with probability at least $\frac{2}{3}$, i.e., the algorithm is a truly perfect $L_2$ sampler, using $\O{\log^2 n}$ bits of space. 
%\end{theorem}
\thmsmallp*
\begin{proof}
By \lemref{lem:l2:output}, \algref{alg:perfect:l2:random} outputs some sample with probability at least $\frac{2}{3}$. 
Conditioned on outputting a sample, \algref{alg:perfect:l2:random} outputs index $i\in[n]$ with probability $\frac{f_i^2}{F_2}$ by \lemref{lem:l2:sample:prob}. 
Since $\mathcal{S}$ maintains $\O{\log n}$ samples, each using $\log n+\log W$ bits, and $\log W=\O{\log n}$, then \algref{alg:perfect:l2:random} uses $\O{\log^2 n}$ bits of space. 
Note that for each update, at most one ordered pair is added to the set $\mathcal{S}$, which may be deleted at some later point in time. 
Thus the amortized update time is $\O{1}$. 
\end{proof}
We now consider truly perfect $L_p$ sampling for $p>2$ on random order streams in the sliding window model. 
For truly perfect $L_p$ sampling on random order streams for integers $p>2$, the idea is to store consecutive blocks of $W^{1-\frac{1}{p-1}}$ elements in the stream, along with their corresponding timestamps. 
We then look for $p$-wise collisions within the block. 
The probability that a fixed group of $p$ positions are all item $i$ is $\frac{f_i}{W}\cdots\frac{f_i-p+1}{W-p+1}$. 

We must therefore again ``correct'' the sampling probability so that the probability of sampling item $i$ is proportional to $f_i^p$. 
To do this, we require the following fact so that we can write $f_i^p$ as a positive linear combination of the quantities $1$, $f_i$, $f_i(f_i-1)$, and so forth. 
\begin{lemma}
\lemlab{lem:falling:factorial}
For an integer $k\ge 0$, let $(x)_k$ denote the falling factorial so that $(x)_0=1$ and $(x)_k=x(x-1)\cdots(x-k+1)$. 
Let $S(n,k)$ denote the Stirling numbers of the second kind, e.g., the number of ways to partition a set of $n$ objects into $k$ non-empty subsets. 
Then for any integer $p\ge0$, we have $x^p=\sum_{k=0}^pS(p,k)(x)_k$. 
\end{lemma}
Thus we can choose to accept the item $i$ with some positive probability if the first $k$ positions of the $p$ positions are all item $i$. 
Otherwise, we can proceed to the first $k+1$ positions, accept the item $i$ with some positive probability if they are all item $i$, and iterate.
We give the algorithm in full in \algref{alg:perfect:lp:random}.

\begin{algorithm}[!htb]
\caption{Truly perfect $L_p$ sampler algorithm for the sliding window model on random order streams for $p>2$.}
\alglab{alg:perfect:lp:random}
\begin{algorithmic}[1]
\Require{A stream of updates $u_1,u_2,\ldots,u_t$, a size $W$ for the sliding window. 
Let each $u_i\in[n]$ represent a single update to a coordinate of the underlying vector $f$.}
\State{$\mathcal{S}\gets\emptyset$}
\State{Store each disjoint block $E$ of $W^{1-\frac{1}{p-1}}$ consecutive elements.}
\For{each ordered $p$-tuple $(v_1,\ldots,v_p)$ of elements in $E$}
\For{$i=1$ to $i=p$}
\If{$v_1=\ldots=v_i$}
\State{Insert $v_1$ into $\mathcal{S}$ with probability $\frac{\alpha_i}{W^{p-1}}$.}
\EndIf
\EndFor
\EndFor
\If{there exists a $p$-tuple in $\mathcal{S}$ with an expired element}
\State{Delete the $p$-tuple from $\mathcal{S}$}
\EndIf
\If{$|\mathcal{S}|>2W^{1-\frac{1}{p-1}}$}
\State{Delete $W^{1-\frac{1}{p-1}}$ elements uniformly at random from $\mathcal{S}$.}
\EndIf
\State{Return a $p$-tuple from $\mathcal{S}$ uniformly at random.}
\end{algorithmic}
\end{algorithm}

\begin{lemma}
\lemlab{lem:lp:sample:prob}
There exist $\alpha_1,\ldots,\alpha_p$ that are efficiently computable so that for each $p$-tuple in $E$, the probability that \algref{alg:perfect:lp:random} samples $j\in[n]$ is $\frac{f_j^p}{W^p}$.
\end{lemma}
\begin{proof}
Let $v_1,\ldots,v_p$ be an ordered $p$-tuple in $E$. 
The probability that $v_1=\ldots=v_p=j$ for some $j\in[n]$ is $\frac{f_j}{W}\cdots\frac{f_j-p+1}{W-p+1}$. 
Similarly the probability that $v_1=\ldots=v_q=j$ for some $q\in[p]$ is $\frac{f_j}{W}\cdots\frac{f_j-q+1}{W-q+1}$. 
Thus by \lemref{lem:falling:factorial}, $\sum_{q=0}^p\alpha_q\frac{f_j}{W}\cdots\frac{f_j-q+1}{W-q+1}=\frac{f_j^p}{W^p}$ for $\alpha_q=S(p,q)\frac{W\cdots(W-q+1)}{W^p}$.
%The probability that $u_{2i-1}$ is $\frac{f_j}{W}$ and $u_{2i-1}$ is sampled with probability $\frac{1}{W}$. 
%Otherwise, $u_{2i}=u_{2i-1}=j$ with probability $\frac{f_j(f_j-1)}{W(W-1)}$. 
%Hence, the probability that \algref{alg:perfect:lp:random} samples $j$ is 
%\[\frac{1}{W}\frac{f_j}{W}+\frac{W-1}{W}\frac{f_j(f_j-1)}{W(W-1)}=\frac{f_j^2}{W^2}.\]
\end{proof}

\begin{lemma}
\lemlab{lem:lp:output}
\algref{alg:perfect:lp:random} outputs some sample with probability at least $\frac{2}{3}$.
\end{lemma}
\begin{proof}
By \lemref{lem:lp:sample:prob}, the probability that \algref{alg:perfect:l2:random} samples some $j\in[n]$ for a $p$-tuple in $E$ is $\frac{F_p}{W^p}$, where $F_p=\sum_{j=1}^n f_j^p$. 
Let $X_1,\ldots,X_m$ be indicator variables so that $X_i=1$ if the $i\th$ $p$-tuple in the window in some fixed ordering produced a sample and $X_i=0$ otherwise and let $X=\sum X_i$. 
Then $m=\left(\frac{W}{W^{1-\frac{1}{p-1}}}\right)\left(W^{1-\frac{1}{p-1}}\right)^p=W^{p-1}$. 
Thus the expected number of samples $X$ inserted into $\mathcal{S}$ is $\Theta\left(\frac{F_p}{W}\right)$. 

Moreover, we have $\Ex{X^2}=\sum\Ex{X_{i_1}\ldots X_{i_p}}$, where $i_1,\ldots,i_p$ are contained in one of the disjoint consecutive block of $W^{1-\frac{1}{p-1}}$ elements. 
By a similar argument as \lemref{lem:l2:output}, $\Ex{X_{i_q}|X_{i_1},\ldots,X_{i_{q-1}}}$ is within a constant multiple of $\Ex{X_{i_q}}$. 
Thus $\Ex{X^2}$ is within a constant factor of $\Ex{X}^2$ so then \algref{alg:perfect:lp:random} outputs some sample with constant probability by the Paley-Zygmund inequality.
\end{proof}

We now show the correctness of our general sliding window algorithm. 
\begin{theorem}
Let $p>2$ be a fixed integer.
There exists a one-pass sliding window algorithm for random order insertion only streams that outputs index $i\in[n]$ with probability $\frac{f_i^p}{\sum_{j=1}^n f_j^p}$ with probability at least $\frac{2}{3}$, i.e., the algorithm is a truly perfect $L_p$ sampler, using $\O{W^{1-\frac{2}{p}}\log n}$ bits of space. 
\end{theorem}
\begin{proof}
By \lemref{lem:lp:output}, \algref{alg:perfect:lp:random} outputs some sample with probability at least $\frac{2}{3}$. 
Conditioned on outputting a sample, \algref{alg:perfect:lp:random} outputs index $i\in[n]$ with probability $\frac{f_i^p}{\sum_{j=1}^n f_j^p}$ by \lemref{lem:lp:sample:prob}. 
Since $\mathcal{S}$ maintains $\O{W^{1-\frac{1}{p-1}}}$ samples, each using $\O{\log n}$ bits for $\log W=\O{\log n}$, then \algref{alg:perfect:lp:random} uses $\O{W^{1-\frac{1}{p-1}}\log n}$ bits of space. 
\end{proof}
We remark that the time complexity of \algref{alg:perfect:lp:random} can be significantly improved for insertion-only data streams. 
Rather than enumerating over all possible $p$-tuples in a block of $\O{W^{1-\frac{1}{p-1}}}$ samples, it suffices to maintain the frequency $g_i$ of each distinct coordinate $i\in[n]$ in the block, in order to simulate the sampling process of the $p$-wise collisions. 
Namely, we observe that although the final step of \algref{alg:perfect:lp:random} can be viewed as selecting a uniformly random element across all tuples ever inserted into $\mathcal{S}$, the same probability distribution holds for any sample uniformly inserted into $\mathcal{S}$ by a single block. 
Thus, we can equivalently instead sample a uniformly random element across all tuples inserted into $\mathcal{S}$ by each block, if the block inserts anything into $\mathcal{S}$ in the original process. 
That is, we can compute the size $k\in[p]$ of the $k$-tuple that represents the uniformly random element inserted into $\mathcal{S}$ by each block, which only requires knowledge of the frequencies $g_i$ for the block. 
Updating the frequency of each coordinates and subsequently updating the marginal probabilities for each index uses $\O{1}$ time per update.  
%and then simulating the number of items that are sampled uses time $\O{\log n}$ by performing a binary search on the number of items that are sampled upon generating a random number. 
%Similarly if the number of samples becomes too large can be simulated, then each added item should be involved in $\O{\log n}$ deletion processes with high probability, so the amortized time is $\O{\log n}$. 
%We can abort the algorithm if an item becomes involved in $\omega(\log n)$ deletion processes, increasing the total failure probability by a negligible amount. 
%Hence, the total amortized update time is $\O{\log n}$. 
Thus, we obtain the following for insertion-only streams:
\thmlargep*
\section{Strict Turnstile Algorithms}
\applab{app:strict:turnstile}
In this section, we show that our techniques can be also be extended to obtaining truly perfect $L_p$ samplers in the strict turnstile model, where updates to each coordinate may be both negative and positive integers but the underlying frequency vector at each point in time consists of non-negative coordinates. 
We assume the magnitude of each update in the stream to be at most $M=\poly(n)$. 

\begin{theorem}[Theorem 4 in~\cite{Ganguly08}]
\thmlab{thm:ksparse:test}
There exists a deterministic streaming algorithm that outputs whether the underlying frequency vector in $\{0,\ldots,M\}^n$ has more than $k$ nonzero coordinates or fewer than $4k$ nonzero coordinates. 
The algorithm uses $\O{k\log(n/k)\log(Mn)}$ bits of space and $\O{\log^2(n/k)}$ amortized time per arriving update.
\end{theorem}
We remark that \thmref{thm:ksparse:test} uses $\O{k\log^2(n/k)}$ time per arriving update as stated in \cite{Ganguly08} due to performing a fast Vandermonde matrix-vector multiplication at each step. 
However, if we instead batch updates by storing the most recent $k$ updates and then amortize performing the fast Vandermonde matrix-vector multiplication over the next $k$ steps, then the amortized running time is $\O{\log^2(n/k)}$ as stated in \thmref{thm:ksparse:test}. 
\begin{theorem}[\cite{GangulyM08, Ganguly08}]
\thmlab{thm:ksparse:recover}
There exists a deterministic streaming algorithm uses space $\O{k\log(Mn)\log(n/k)}$ and update time $\polylog(n/k)$ and recovers an underlying $k$-sparse frequency vector in $\{-M,\ldots,M\}^n$. 
\end{theorem}
We similarly remark that \thmref{thm:ksparse:recover} as stated in \cite{GangulyM08, Ganguly08} uses $\O{k\log(n/k)}$ time per arriving update due to maintaining counters tracking the degrees of $\O{k\log(n/k)}$ vertices on the right side of a regular bipartite approximate lossless expander graph, where each vertex on the left corresponds to a coordinate of the universe and has degree $\polylog(n/k)$. 
Moreover, the lossless expander graph is explicit, so computation of the neighborhood of each vertex uses $\polylog(n/k)$ time~\cite{GuruswamiUV09}. 
Thus if we again batch updates by storing the most recent $k$ updates and then amortize updating the degrees of the $\O{k\log(n/k)}$ vertices on the right side of the graph, then the amortized running time is $\polylog(n/k)$ as stated in \thmref{thm:ksparse:recover}.  

By setting $k=2\sqrt{n}$ in \thmref{thm:ksparse:recover}, we can then form a set $T$ of up to $8\sqrt{n}$ unique nonzero coordinates in the frequency vector, analogous to \algref{alg:F0:simple}. 
Similarly, setting $k=2\sqrt{n}$ in \thmref{thm:ksparse:test} tests whether the sparsity of the vector is at most $8\sqrt{n}$, in which case a random coordinate of $T$ can be returned, or the sparsity of the vector is at least $2\sqrt{n}$, in which case an element of a random subset $S$ of size $2\sqrt{n}$ as defined in \algref{alg:F0:simple} will be a nonzero coordinate of the vector with constant probability. 
Thus we obtain the same guarantees as \thmref{thm:truly:perfect:F0} for a strict turnstile stream. 

\begin{theorem}
\thmlab{thm:truly:perfect:F0:turnstile}
Given $\delta\in(0,1)$, there exists a truly perfect $F_0$ sampler on strict turnstile streams that uses $\O{\sqrt{n}\log^2 n\log\frac{1}{\delta}}$ bits of space and $\polylog n\log\frac{1}{\delta}$ update time and succeeds with probability at least $1-\delta$.  
\end{theorem}

We now show that we can get truly perfect $L_p$ samplers in strict turnstile streams using $\O{n^{1-1/p+\gamma}\log n}$ space over $\O{1/\gamma}$ constant number of passes, where $\gamma>0$ is a trade-off parameter. 
This shows a separation for $L_p$ samplers between general turnstile streams and strict turnstile streams.

We first consider a truly perfect $L_1$ sampler. 
The idea is just to partition the universe $[n]$ into $n^{\gamma}$ chunks so that for each $i\in[\gamma]$, the $i$-th chunk corresponds to the coordinates between $(i-1)n^{\gamma}+1$ and $in^{\gamma}$. 
Our algorithm stores the sum of counts on each chunk in the first pass, which requires $\O{n^{\gamma}\log n}$ bits of space. 
We then sample a chunk proportional to this sum, so that the total universe size has decreased to be size $n^{1-\gamma}$. 
In the subsequent passes, we recursively partition the decreased universe into $n^{\gamma}$ chunks as before and repeat. 
Thus after repeating $\O{1/\gamma}$, we sample a single coordinate under the desired distribution. 

For general $p$, note that our algorithm for insertion-only is actually just sampling according to frequency. 
Thus we can use $\O{1/\gamma}$ passes as in the previous paragraph to do this, and then we can simulate our old algorithm with these samples. 
As the insertion-only algorithm uses $\O{n^{1-1/p}}$ samples, we require $\O{n^{1-1/p + \gamma}\log n}$ bits of space in $\O{1/\gamma}$ passes. 
Upon obtaining the samples, we also need to deterministically compute a number $Z$ such that $\|f\|_{\infty}\le Z\le\|f\|_{\infty}+\frac{m}{n^{1-1/p}}$ for the frequency vector $f$. 
To do this, we can partition the universe $[n]$ into $n^{1-1/p+\gamma}$ chunks of size $n^{1/p-\gamma}$. 
There can be at most $n^{1-1/p}$ chunks whose coordinate sums are at least $\frac{m}{n^{1-1/p}}$. 
Thus, we have reduced the number of possible universe items by a factor of $1/n^{\gamma}$. 
Hence after $\O{1/\gamma}$ passes repetitions, we find any such item $\|f\|_{\infty}$ and thus obtain a deterministic estimate $Z$ such that $\|f\|_{\infty}\le Z\le\|f\|_{\infty}+\frac{m}{n^{1-1/p}}$. 

More generally, we have the following reduction:
\thmperfectmultipass*
\end{document}